\documentclass{article}

\usepackage{algorithmic}
\usepackage{algorithm}
\usepackage{amssymb}
\usepackage{amsmath}
\usepackage{amsfonts}
\usepackage{booktabs}
\usepackage[hmargin=.75in,vmargin=1in]{geometry}
\usepackage{hyperref}
\usepackage{microtype}
\usepackage{pgf}
\usepackage{pgfarrows}
\usepackage{pgfnodes}
\usepackage{pgfautomata}
\usepackage{pgfheaps}
\usepackage[loose]{units}
\usepackage{wrapfig}
\usepackage{xspace}
\usepackage[misc]{ifsym}

\newtheorem{theorem}{Theorem}[section] 
\newtheorem{lemma}[theorem]{Lemma}

\newtheorem{corollary}[theorem]{Corollary}

\newtheorem{definition}[theorem]{Definition}
\newtheorem{observation}[theorem]{Observation}
\newtheorem{fact}[theorem]{Fact}
\newtheorem{conclusion}[theorem]{Conclusion}
\newcommand{\CirculantGraph}{\small{\textsf{CirculantGraph}}\xspace}

\newcommand{\Lang}[1]{\ifmmode{\text{\textsc{#1}}}\else\textsc{#1}\fi}

\newcommand{\lelex}{\ensuremath{\le_{\text{lex}}}}

\newcommand{\Nat}{\ensuremath{\mathbb{N}}}

\newcommand{\cB}{{\ensuremath{\cal B}}}
\newcommand{\cC}{{\ensuremath{\cal C}}}

\newcommand{\cG}{{\ensuremath{\cal G}}}

\newcommand{\cP}{{\ensuremath{\cal P}}\xspace}

\newcommand{\cNP}{{\ensuremath{\cal NP}}\xspace}
\newcommand{\cBH}{{\ensuremath{\cal BH}}}
\newcommand{\cOC}{{\ensuremath{\cal OC}}}

\newcommand{\CoV}{\text{CoV}}

\newcommand{\splitExt}[1]{{\ensuremath{\Gamma_{\text{split}}(#1)}}}
\newcommand{\splitExtP}[2]{{\ensuremath{\Gamma_{\text{split}}^{#2}(#1)}}}
\newcommand{\splitExtB}{{\ensuremath{\Gamma_{\text{split}}}}}

\newcommand{\parExt}[1]{{\ensuremath{\Gamma_{\text{par}}(#1)}}}
\newcommand{\parExtP}[2]{{\ensuremath{\Gamma_{\text{par}}^{#2}(#1)}}}
\newcommand{\parExtB}{{\ensuremath{\Gamma_{\text{par}}}}}
\newcommand{\parCla}[1]{{\ensuremath{[#1]_{\text{par}}}}}
\newcommand{\chaExt}[1]{{\ensuremath{\Gamma_{\text{chain}}(#1)}}}
\newcommand{\chaExtP}[2]{{\ensuremath{\Gamma_{\text{chain}}^{#2}(#1)}}}
\newcommand{\chaExtB}{{\ensuremath{\Gamma_{\text{chain}}}}}

\newcommand{\chaParExt}[1]{{\ensuremath{\Gamma_{\text{chain, par}}(#1)}}}
\newcommand{\chaParExtP}[2]{{\ensuremath{\Gamma_{\text{chain, par}}^{#2}(#1)}}}

\newcommand{\splitParExt}[1]{{\ensuremath{\Gamma_{\text{split, par}}(#1)}}}
\newcommand{\splitParExtP}[2]{{\ensuremath{\Gamma_{\text{split, par}}^{#2}(#1)}}}
\newcommand{\splitParExtB}{{\ensuremath{\Gamma_{\text{split, par}}}}}

\newcommand\qedsymbol{\textcolor{darkgray}{\ensuremath{\blacktriangleleft}}}

\newenvironment{proof}{\par\smallskip\noindent\textit{Proof.}\ }{\quad\hfill\qedsymbol\par\smallskip}

\let\origthebibliography=\thebibliography
\renewcommand\thebibliography[1]{\small\origthebibliography{#1}\parskip0pt\itemsep0pt}

\newcommand{\UEdge}[1]{\ensuremath{\{#1\}}}

\newcommand{\CIrreducible}{$\alpha$-Critical\xspace}

\newcommand{\Reducible}{uncritical\xspace}
\newcommand{\Irreducible}{$\alpha$-critical\xspace}
\newcommand{\VCO}{VC-Overlap\xspace}
\newcommand{\VCOO}{$1$-VC-Overlap\xspace}
\newcommand{\VCR}{\Reducible}
\newcommand{\VCI}{\Irreducible}
\newcommand{\NeEq}{neighbor-equivalent\xspace}

\newcommand{\eg}{e.g.,\xspace} 
\newcommand{\paragraphbf}[1]{{\scshape{#1}}\xspace}

\newcommand{\set}[1]{\ensuremath{\mathcal{#1}}\xspace}

\newcommand{\V}{\ensuremath{V}\xspace}
\newcommand{\E}{\ensuremath{E}\xspace}
\newcommand{\G}{\ensuremath{G}\xspace}

\newcommand{\cover}{\ensuremath{\set{C}}\xspace}

\begin{document}


\title{Critical Graphs for Minimum Vertex Cover}
\author{{
Andreas Jakoby, 
Naveen Kumar Goswami,
Eik List,
Stefan Lucks} \\
Faculty of Media, 
Bauhaus-Universit\"at Weimar, Bauhausstr. 11, D-99423, Weimar, Germany\\
{\small\texttt{<firstname.lastname>@uni-weimar.de}}}

\maketitle


\begin{abstract}
An $\alpha$-critical graph is an instance where the deletion of any element
would decrease some graph's measure $\alpha$. Such instances have shown to be
interesting objects of study for deepen the understanding of optimization
problems.

This work explores critical graphs in the context of the Minimum-Vertex-Cover
problem. We demonstrate their potential for the generation of larger graphs with
hidden a priori known solutions. Firstly, we propose a parametrized graph-
generation process which preserves the knowledge of the minimum cover. Secondly,
we conduct a systematic search for small critical graphs. Thirdly, we illustrate
the applicability for benchmarking purposes by reporting on a series of
experiments using the state-of-the-art heuristic solver NuMVC.
\end{abstract}

\vspace{1em}
\noindent\textbf{Keywords:}
 {\small  critical graphs, minimum vertex cover, graph generation, 
benchmark generator} 

\section{Introduction}
\label{sec:introduction}

\paragraphbf{The Minimum Vertex Cover Problem.}
A \emph{vertex cover} \cover for a given graph $\G = (\V, \E)$ defines a subset
of vertices $\cover \subseteq \V$ such that every edge in \E is incident to at
least one vertex in \cover. A \emph{minimum vertex cover} (MVC) is a vertex
cover with the smallest possible size. The task of finding a minimum vertex
cover in a given graph is a classical \cNP-hard optimization
problem~\cite{garey:1979}, and its decision version one the 21 original
\cNP-complete problems listed by Karp~\cite{karp:1972}. While the construction
of a maximal matching yields a trivial approximation-2 algorithm, it is 
\cNP-hard to approximate MVC within any factor smaller than
1.3606~\cite{dinur:2005} unless $\cP = \cNP$, according to the Unique Game
Conjecture; though, one can achieve an approximation ratio of $2 -
o(1)$~\cite{karakostas:2005}.

The MVC problem is strongly related to at least three further
\cNP-hard problems. Finding a minimum vertex cover is equivalent to
finding a Maximum Independent Set (MIS), i.~e. a subset of vertices wherein no
pair of vertices shares an edge. An MIS problem instance can again be
transformed into an instance of the Maximum-size Clique (MC) problem; moreover,
there is a straight-forward reduction of a (binary) Constraint Satisfaction
Problem (CSP) to an MIS problem. In practice, the MVC problem plays an important
role in network security, industrial machine assignment, or facility
location~\cite{kavalci:2014,pirzada:2007,zhang:2006}. Furthermore, MVC
algorithms can be used for solving MIS problems \eg in the analysis of social
networks, pattern recognition, and alignment of protein sequences in
bioinformatics~\cite{pirzada:2007,stege:2000}.

\paragraphbf{\CIrreducible Graphs.}
Graphs are called critical if they are minimal with regards to a
certain measure. More precisely, an edge of a given graph $G$ is called a
\emph{critical element} iff its deletion would decrease the measure. $G$ is then
called \emph{edge-critical} (or simply critical) iff every edge is a critical
element. The concept of critical graphs is mostly used in works on the chromatic
number. Though, it can also be of significant interest for the MVC problem,
where we call a graph critical iff the deletion of any edge would decrease the
size of the minimal cover. This concept has been introduced by 
Erd\"{o}s and Gallai \cite{erdos:1961} in 1961 using the term $\alpha$-critical
where $\alpha$ denotes the certain measure on a graph. E.g., let $\alpha$
denote a function determining the size of a maximum independent set, then 
we call a graph $G=(V,E)$ $\alpha$-critical if for any edge $e\in E$ it
holds that $\alpha(G)<\alpha(G')$ where $G'=(V,E\setminus\{e\})$. Since 
the size of a minimum vertex cover of a graph $G=(V,E)$ is given by 
$|V|-\alpha(G)$, this definition and the correlated results 
apply directly to the MVC problem. 
The insights of studying such $\alpha$-critical graphs could help to
deepen our understanding on the complexity of the Vertex Cover problem or to
find more efficient solvers.
Useful summaries on $\alpha$-critical graphs can be found in \cite{lovasz:1993,joret:2007}.

\paragraphbf{Randomized Graph Generation.}
Critical graphs can further serve as the base for constructing larger graphs.
Since small critical instances possess an easily determined cover size, a
parametrized graph-generation process that preserves the criticality could
create large instances while maintaining the knowledge about the solution. A
potential application for such graphs could be, e.g., the dedicated generation
of particularly hard instances for benchmarking purposes. Following a series of previous graph-%
generation models~\cite{smith:1996,xu:2000}, the idea for generating such graphs
for the minimum-vertex cover problem had been introduced by Xu and
Li~\cite{xu:2003,DBLP:journals/tcs/XuL06} and revisited in~\cite{xu:2007}.
During the past decade, Xu's BHOSLIB suite~\cite{xu:2005} has established as a
valuable benchmark suite for the evaluation of MVC, MIS, MC, and CSP solvers.

\paragraphbf{Contribution.}
This work studies critical graphs for the Minimum Vertex Cover problem. First,
we propose a (not necessarily efficiently implementable) graph-generation
process which can create all possible graphs while preserving the knowledge of
the minimum cover size. To implement this process efficiently, we
restrict it to a certain set of extensions that enlarge a critical graph while
maintaining the criticality. Second, we systematically search for small critical
instances. As a useful observation, we show that, if all critical graphs for the
MVC problem were known, our restricted process could efficiently generate all
possible graphs. Third, we illustrate the applicability of instances generated
by our process for benchmarking. We report on a series of experiments with a
state-of-the-art heuristic solver NuMVC~\cite{cai:2013} on examplary instances
that were generated by our randomized process.

\paragraphbf{Outline.}
In the following, Section~\ref{sec:irreducible} defines \Irreducible graphs for
the Vertex-Cover problem. Section~\ref{sec:generating} describes our approach
for generating hard random graphs from \Irreducible graphs.
Section~\ref{sec:extensions} presents our used extensions.
Section~\ref{sec:results} details the results of our
experiments and Section~\ref{sec:conclusion} concludes.


\section{\CIrreducible Graphs for the Minimum-Vertex-Cover Problem}
\label{sec:irreducible}

\begin{definition} 
  A connected graph $G = (V,E)$ is edge-\Reducible{} (\Reducible{} hereafter)
  according to an optimization problem $P$ on graphs iff there exists an edge
  $e\in E$ such that every solution for $P$ at $G'=(V,E\setminus\{e\})$ is a
  solution for $P$ at $G$. A connected graph is edge-\Irreducible{}
  (\Irreducible{} hereafter) according to an optimization problem $P$ iff it is
  not \Reducible.
\end{definition}

For the vertex cover problem, this definition implies:

\begin{observation}
  A connected graph $G=(V,E)$ is \Irreducible{} according to the vertex cover
  problem (in short: \VCI) iff deleting any edge reduces the minimum cover size.
\end{observation}

There are several simple \Irreducible graphs.

\begin{fact}
  \label{f:CycClique}
  Cliques and cycles of odd length are \VCI.
\end{fact}
\begin{proof}
Note that the size of a minimum vertex cover of a clique 
$C_k=(V,E)$ of size $k$
is $k-1$. If we delete an edge $e=\UEdge{u,v}\in E$ 
from $C_k$, then
$V\setminus\{u,v\}$ is a vertex cover of size $k-2$ for
$(V,E\setminus\{e\})$. Thus, deleting any edge from $C_k$ 
gives a graph with
reduced vertex cover size. Cliques are \Irreducible.

Note that the size of a minimum vertex cover of a cycle 
$G_k=(V,E)$ of odd
length $2k+1$ is $k+1$. If we delete an edge $e\in E$ 
from $G_k$ the resulting
graph is a simple chain of $2k+1$ vertices. W.l.o.g., let
$V=\{v_1,\ldots,v_{2k+1}\}$ and 
$E\setminus\{e\}=\{\UEdge{v_i,v_{i+1}}|1\le i\le 2k\}$. 
Then $\{v_{2i}|1\le i\le k\}$ gives a vertex cover of size 
$k$ for
$(V,E\setminus\{e\})$. Thus, deleting any edge from 
$G_k$ gives a graph with
reduced vertex cover size. 
Cycles of odd length are \Irreducible.
\end{proof}

For an overview on $\alpha$-critical graphs see for example \cite{lovasz:1993,joret:2007}. 


Recall that a perfect matching of a graph is defined as follows:

\begin{definition}
Let $G = (V , E)$ be an undirected graph and $E'\subseteq E$. 
Then $E'$ is
called a matching if for every pair $e_1, e_2 \in E'$ 
of edges either $e_1 =
e_2$ or $e_1 \cap e_2 = \emptyset$, i.e. if two 
edges in $E'$ are different then
the adjacent pairs of nodes are disjoint.

A matching $E'$ is a perfect matching if each node of 
$G$ is incident to an edge
in $E'$.
\end{definition}

Analyzing graphs with a perfect matching, one can show that cycles of even
length are \Reducible. Now, we can easily prove the following observations:

\begin{observation}
  \label{ops:perfect_matching}
  Let $G=(V,E)$ be a connected undirected graph that has a perfect matching
  $E'$. The minimal size of a vertex cover is at least $|V|/2$. Moreover, if the
  minimal size of a vertex cover is exactly $|V|/2$, then either $E=E'$, i.e.,
  $G$ is its own perfect matching, or $|E|>|E'|$, and then $G$ is \VCR.
\end{observation}
\begin{proof}
Observe that $|V|/2=|E'|$, and $|E'|$ vertices are needed to cover all edges in
$E'$. This proves the first claim.

For the second claim, assume that there exists a minimum vertex cover of size
$|V|/2$ of $G$ and $|E|>|E'|$. This implies at least one edge $\UEdge{u,v} \in
E$, with $\UEdge{u,v} \not\in E'$. Deleting $\UEdge{u,v}$ from $G$ gives a
smaller graph $G'$. $E'$ is a perfect matching, not only of $G$, but also of
$G'$. From the first claim, we know the minimal size of a vertex cover of $G'$
is at least $|V|/2$. Since removing $\UEdge{u,v}$  from $G$ does not change the
size of a minimum vertex cover, $G$ is \Reducible.
\end{proof}

Since there exists a perfect matching for cycles of even length,
Observation~\ref{ops:perfect_matching} implies:

\begin{corollary}
  \label{cor:evencycles}
  Cycles of even length $2k$ with $k\ge 2$ are \VCR.
\end{corollary}

\begin{observation}
  \label{ops:2endpoints}
  Let $G = (V,E)$ be a connected undirected graph and let $U\subset V$ be a
  minimum vertex cover for $G$. Assume that there exists an edge $e\in E$ such
  that both endpoints of $e$ are in the cover $U$. Then, either $G$ is
  \Reducible{} or there exists a minimum vertex cover $U'$ for $G$ such that
  only one of the endpoints of $e$ is in the cover $U'$.
\end{observation}
\begin{proof}
Let $G=(V,E)$ be a connected undirected graph and let 
$U\subset V$ be a minimum
vertex cover for $G$. Let $e=\UEdge{v_1,v_2}\in E$ such that 
both endpoints of $e$
are in $U$. Assume that $G$ is \Irreducible. Then 
$G'=(V,E\setminus\{e\})$ has a
minimum vertex cover $U''$ such that non of the two endpoints 
of the deleted
edge $e$ is in the cover; otherwise, $U''$ has already been 
a cover for $G$ --
contradicting the assumption that $U$ is a minimum vertex cover.

Since $U''$ covers $G'$ both sets $U''\cup\{v_1\}$ 
and $U''\cup\{v_2\}$ cover
$G$ and include only one of the endpoints of $e$.
\end{proof}

Given a graph $G=(V,E)$ and a minimum cover $U \subseteq V$, such that no
$\UEdge{u,v} \in E$ exists with $u,v \in U$, then the graph is bipartite. In
that case, the minimum cover is one side of the bipartite decomposition of the
vertex set.

Furthermore, one can show that \VCI graphs have to be $2$-connected.

\begin{theorem}
  \label{th:2connected}
  Let $G = (V, E)$ be a graph with an articulation vertex $u$, then $G$ is \VCR.
\end{theorem}

To prove Theorem~\ref{th:2connected} we have to start with showing some 
useful lemmas.
Recall that a vertex is called an articulation vertex of a connected graph if its removal will disconnect the graph. 

\begin{figure}[h]
  \begin{center}
    \scalebox{0.25}{
      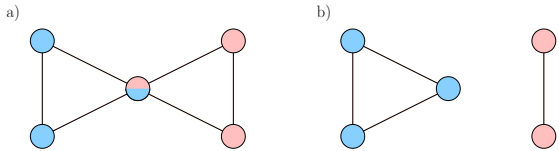
    }
  \end{center}
  \vspace{-1em}
  \caption{Two graphs with a minimum vertex cover size of 3. Since the 
  second graph is a subgraph of the first one, the first is \VCR{}.}
  \label{fig:2x3Cycle}
\end{figure}

\begin{lemma}
\label{l:articulationVertex01}
Let $G=(V,E)$ be a graph with a articulation vertex $u$ that subdivides $G$ into
two subgraphs $G_1$ and $G_2$ (both include $u$). For a vertex cover $C$ of $G$,
let $C_1$ denote the subset of $C$ that denotes a cover for $G_1$ and let $C_2$
denote the subset of $C$ that denotes a cover for $G_2$. If $C$ is a minimum
vertex cover for $G$ then either $C_1$ is a minimum vertex cover for $G_1$, or
$C_2$ is a minimum vertex cover for $G_2$, or both sets are minimum vertex
covers for the two corresponding subgraphs.
\end{lemma}
\begin{proof}
We prove this observation by a contradiction.
Let $C$ be an optimal vertex cover of $G$, $C_1$ be the subset 
of $C$ that denotes a cover for $G_1$, and $C_2$ be the subset 
of $C$ that denotes a cover for $G_2$. 

\begin{figure}[h]
  \begin{center}
    \scalebox{.3}{
      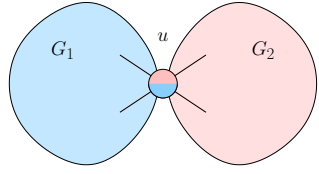
    }
  \end{center}
  \vspace{-1em}
  \caption{Articulation vertex $u$ with two subgraphs $G_1$ and $G_2$.}
  \label{fig:articVert}
\end{figure}

Assume that neither $C_1$ nor $C_2$
is an optimal vertex cover for $G_1$ or $G_2$, respectively.
Let $C_i'$ be an optimal vertex cover for $G_i$, then $C'=C_1'\cup C_2'$
is a vertex cover for $G$ of size $|C'|\le |C_1'|+|C_2'|$. According to our
assumption $C'$ is not an optimal vertex cover; thus, we have
\begin{eqnarray*}
|C_1'|+|C_2'| & \ge & |C'| \ \ > \ \ |C| \\ 
              & \ge & |C_1|+|C_2|-1 \ \ > \ \ |C_1'|+|C_2'|+1
\end{eqnarray*}
-- a contradiction.
\end{proof}

\begin{lemma}
\label{l:articulationVertex02}
Let $G=(V,E)$ be a graph with an articulation vertex $u$ that 
subdivides $G$ into two subgraphs $G_1$ and $G_2$ 
(both include $u$). If there exists a minimum vertex cover $C_1$ 
for $G_1$ with $u\in C_1$, then $G$ is \VCR.
\end{lemma}
\begin{proof}
Assume that there exists a minimum vertex cover $C_1$ 
for $G_1$ with $u\in C_1$. Furthermore, assume that $G$ is \VCI. 

Let $G'$ be the subgraph of $G$ that does not contain any edge $\{u,v\}$ that
connects $u$ with another vertex of $G_2$. Since $G$ is \VCI, there exists
vertex cover $C'$ for $G'$ that is smaller than the minimum vertex cover size of
$G$. If we now replace the part of $C'$ that denotes a vertex cover for $G_1$ by
$C_1$, we get a vertex cover $C$ for $G'$ of the same size as $C'$ that contains
$u$. Thus, $C$ is also a vertex cover for $G$ -- contradicting our assumption
that the minimum vertex cover size of $G$ is larger than $|C'|=|C|$ and thus $G$
is \VCI.
\end{proof}

\begin{conclusion}
\label{conc:artic1}
Let $G=(V,E)$ be a \VCI{} graph with a articulation vertex $u$. Then every 
minimum vertex cover of $G$ does not include $u$.
\end{conclusion}

\begin{conclusion}
\label{conc:artic2}
Let $G=(V,E)$ be a \VCI{} graph with a articulation vertex $u$. Then every 
neighbor of $u$ in $G$ has to be an element of every 
minimum vertex cover of $G$.
\end{conclusion}

Now we can prove Theorem~\ref{th:2connected}:

\begin{proof}[of Theorem~\ref{th:2connected}]
For the contrary, assume that $G$ is \VCI.
Let $v$ be a neighbor of $u$ in $G$ and let $G'=(V,E\setminus\{\{u,v\}\})$ be the
subgraph of $G$ after removing the edge $\{u,v\}$, then we can distinguish
between the following two cases:
\begin{itemize}
\item The minimum vertex cover $C'$ of $G'$ is smaller than the minimum vertex cover for $G$. 
But this implies that $C'\cup\{u\}$ gives a vertex cover for $G$ of size $|C|$ -- contradicting 
conclusion~\ref{conc:artic1}.
\item The minimum vertex cover $C'$ of $G'$ is of the same size as 
the minimum vertex cover for $G$. 
But this implies that $G$ is \VCR.
\end{itemize}
\end{proof}

Another useful observation is the following:

\begin{observation}
  \label{obs:elementsOfC}
  Let $G = (V,E)$ be a \VCI graph. Then, for every vertex $u \in V$, there
  exists a minimum vertex cover $C$ of $G$ with $u \in  C$.
\end{observation}

\begin{proof}
Let $v$ be a neighbor of $u$ in $G$ and let $G'=(V,E\setminus\{\{u,v\}\})$ be
the subgraph of $G$ after removing the edge $\{u,v\}$, then (because $G$ is
\VCI) the minimum vertex cover $C'$ of $G'$ is smaller than the minimum vertex
cover for $G$. Thus, $C'\cup\{u\}$ is a minimum vertex cover of $G$.
\end{proof}

\section{Generating Larger Graphs from \CIrreducible Instances}
\label{sec:generating}

One relevant application of critical graphs is the construction of larger
graphs. For this purpose, we need a (randomized) generation process which (1)
allows to construct all possible graphs, and (2) preserves the knowledge of a
hidden solution, i.e. about the minimum vertex cover. This section presents such
a generation process.

\begin{definition}
  \label{def:random_processes}
  Let $\cB = \{B_1, B_2, \ldots\}$ be a set of graphs where each graph is given
  by a triple $B_i = (U, V, E)$ with two disjoint sets of vertices $U$ and $V$
  such that $U$ gives a minimum vertex cover for the graph $(U \cup V, E)$.
  Define the following random processes
  \begin{itemize}
    \item Given $\cB$ and $\ell, m, n \in \Nat$ for the vertex cover size
    $\ell$, an upper bound $m$ for the number of edges, and an upper bound $n$
    for the number of vertices, then $\cG_{\cB, \ell, m, n}^1$ is a random
    variable that uniformly at random gives a collection $S_1, \ldots, S_k$ of
    elements of $\cB$ (some elements of $\cB$ may repeat) with $S_i
    = (U_i, V_i, E_i)$ s.t.
    $$
      \left|\bigcup_{i = 1}^k U_i\right| = \ell \text{ and }
      \left|\bigcup_{i = 1}^k V_i\right|+\ell \le n\text{ and } 
      \left|\bigcup_{i = 1}^k E_i\right| \le m\ .
    $$
    Let $C(S_1, \ldots, S_k) = 
    (\bigcup_{i = 1}^k U_i, \bigcup_{i = 1}^k V_i, \bigcup_{i = 1}^k E_i)$.
    \item Given a triple $B = (U, V, E)$ with $|V|+|U|\le n$, then
    $\cG_{n}^2(B)$ will be the triple $(U, V \cup V', E)$ where $V'$ denotes a
    set of $n-|U|-|V|$ new vertices ($V'\cap (U \cup V) = \emptyset$).
    \item Given a triple $B = (U, V, E)$ with $|E| \le m$, let $\cG_{m}^3(B)$ be
    a random variable that uniformly at random adds $m - |E|$ new edges $E'
    \subset U \times (U \cup V)$ to $B$.
  \end{itemize}
\end{definition}

Since none of the defined processes reduces the cover size, we can conclude:

\begin{theorem}
  \label{the:randProc-1}
  Let $\cB = \{B_1, B_2, \ldots\}$ be a set of graphs, where each graph is given
  by a triple $B_i = (U, V, E)$ with two disjoint sets of vertices $U$ and $V$
  such that $U$ gives a minimum vertex cover for $(U\cup V, E)$. Then, for every
  random graph $(U', V', E')$ in the range of 
  $
    \cG_{m}^3(\cG_{n}^2(C(\cG_{\cB, \ell, m, n}^1))),
  $, $U'$ is a minimum vertex cover for $(U'\cup V', E')$ of size $\ell$, and
  the graph $(U'\cup V', E')$ has $n$ vertices and $m$ edges.
\end{theorem}
\begin{proof}
The second claim follows directly from the definition 
of the two processes $\cG_{m}^3$ and $\cG_{n}^2$.

Furthermore, the process $\cG_{m}^3$, $\cG_{n}^2$, $C$, or 
$\cG_{\cB,\ell,m,n}^1$ adds an edge to the resulting graph 
that is not covered by the
cover sets that are given by the subgraphs 
$B_i=(U,V,E)\in \cB$ that are chosen by the
first process $\cG_{\cB,\ell,m,n}^1$. 
Thus, the covers are maintained.

Thus, it suffices to prove that the defined processes 
do not reduce the cover size.

Let $G=(V,E)$ be an arbitrary graph and let 
$G'=(V',E')$ be a subgraph of $G$,
i.e. $V'\subseteq V$ and $E'\subseteq E\cap(V'\times V')$. 
Furthermore, let $C$
be a vertex cover of $G$. Then $C\cap V'$ is a vertex cover of $G'$.

Now assume that there exists a graph $(U',V',E')$ in the range of
$$\cG_{m}^3(\cG_{n}^2(C(\cG_{\cB,\ell,m,n}^1)))$$ such
that $(U'\cup V', E')$ has a cover $C$ of size $<\ell$. 
Since the subgraphs
$S_i$ chosen by the first process are vertex disjoint 
there has to be such a
chosen subgraph $S_i=(U_i,V_i,E_i)$ such that 
$C_i=C\cap (U_i\cup V_i)$ is
strictly smaller than $|U_i|$ and $C_i$ is a vertex cover 
for $(U_i\cup V_i,E_i)$. This contradicts the assumption that 
$U_i$ is a minimum vertex cover
for $(U_i\cup V_i, E_i)$.
\end{proof}

Thus, the three processes do not affect our a-priori knowledge of the minimum
cover size. It remains to show that any graph can be constructed by the three
processes. This observation follows by analyzing the reverse processes:

\begin{theorem}
  \label{the:randProc-2}
  Let $\cB = \{B_1, B_2, \ldots\}$ with $B_i = (U, V, E)$ be a set of all
  \Irreducible{} graphs $G_i = (U \cup V, E)$, where $U$ and $V$ are disjoint
  and $U$ gives a minimum vertex cover for the graph $G_i$. Then the range of
  $\cG_{m}^3(\cG_{n}^2(C(\cG_{\cB, \ell, m, n}^1)))$ determines the set of all
  graphs of $n$ vertices, $m$ edges, and minimum vertex cover size $\ell$.
\end{theorem}
\begin{proof}
For the contrary, assume that there exists a graph $G=(V,E)$ 
of $n$ vertices, $m$ edges, and minimum vertex cover size 
$\ell$ that is not in the range of
$\cG_{m}^3(\cG_{n}^2(C(\cG_{\cB,\ell,m,n}^1)))$. Let
$U\subset V$ be a vertex cover of $G$ of size $\ell$.
\begin{itemize}
  \item Then either $G$ is connected and 
  $G\not\in\cB$ and therefor $G$ is
  \Reducible{} or
  \item there exists a connected component $G'$ 
  of $G$ with $n'$ vertices, $m'$
  edges, and minimum vertex cover size $\ell'$ 
  where $G'$ is not in the range of
  $$\cG_{m'}^3(\cG_{n'}^2(C(\cG_{\cB,\ell',m',n'}^1)))$$ and
  therefore $G'\not\in\cB$ and thus $G'$ is \Reducible{}.
\end{itemize} 
In the following we will restrict ourselves to the first case. The second case
follows analogously by focusing on the components $G'$ that are not in the range
of $\cG_{m'}^3(\cG_{n'}^2(C(\cG_{\cB,\ell',m',n'}^1)))$. 

Since $G$ is \Reducible{}, there exists an edge $e$ that 
can be deleted from the
corresponding graph without reducing the minimum size of 
a vertex cover. Since
at least one the endpoints of $e$ is in $U$ 
this edge can be generated by the
process $\cG_{m}^3$. Hence, if $G$ is not in the range of
$$\cG_{m}^3(\cG_{n}^2(C(\cG_{\cB,\ell,m,n}^1)))$$ then
$G''=(V,E\setminus\{e\})$ is not in the range of
$$\cG_{m-1}^3(\cG_{n}^2(C(\cG_{\cB,\ell,m-1,n}^1)))\ .$$
This reduction step reduces the graph until we have 
deleted all edges from the
graph without changing the initial vertex cover $C$. 
Thus, $C$ has to be empty
and therefore also the initial set of edges has to be empty. 
Hence, $G$ has to
be a set of isolated vertices and $\ell=0$. But this graph 
is in the range of
$$\cG_{0}^3(\cG_{n}^2(C(\cG_{\cB,0,0,n}^1)))$$ and will be
generated by the subprocess $\cG_{\cB,0,0,n}^2$ -- a contradiction.
\end{proof}

\section{Circulant Graphs}
\label{sec:search}

We will now discuss graph classes that can be seen as a
generalization of cycles and as extensions of cliques.

\begin{definition}
A circulant graph $\CirculantGraph(n, L)$ with $n\in \Nat$ and 
$L\subseteq \{1,\ldots,\rceil n/2\lceil\}$ is an undirected graph with $n$ 
vertices $v_0,\ldots,v_{n-1}$ where each vertex $v_i$ is adjacent to
both vertices $v_{(i+j)\bmod n}$ and $v_{(i-j)\bmod n}$ for all $j\in L$.
\end{definition}

To determine critical graphs, we analyzed circulant graphs of degree 4, 
i.e. $\CirculantGraph(n, L)$ with $n \in [2,80]$, $L=\{1,j\}$, and 
$j \in [2, 20]$. We identified $121$ critical graphs with degree $4$ 
in this range which are visualized in 
Figure~\ref{fig:critical-circulant-graphs-2}.
Furthermore, we implemented a search for all critical 
circulant graphs of degree 6, i.e. for $\CirculantGraph(n, L)$ 
with $|L|=3$. For $n \in [4,60]$, $L=\{1,i,j\}$, and $i, j \in [2, 20]$.
We determined in total $427$ critical graphs within this range. 
These are listed in Table~\ref{fig:critical-circulant-graphs-2} in the Appendix. 

\begin{figure}[h]
  \begin{center}
    \scalebox{0.66}{
      \includegraphics{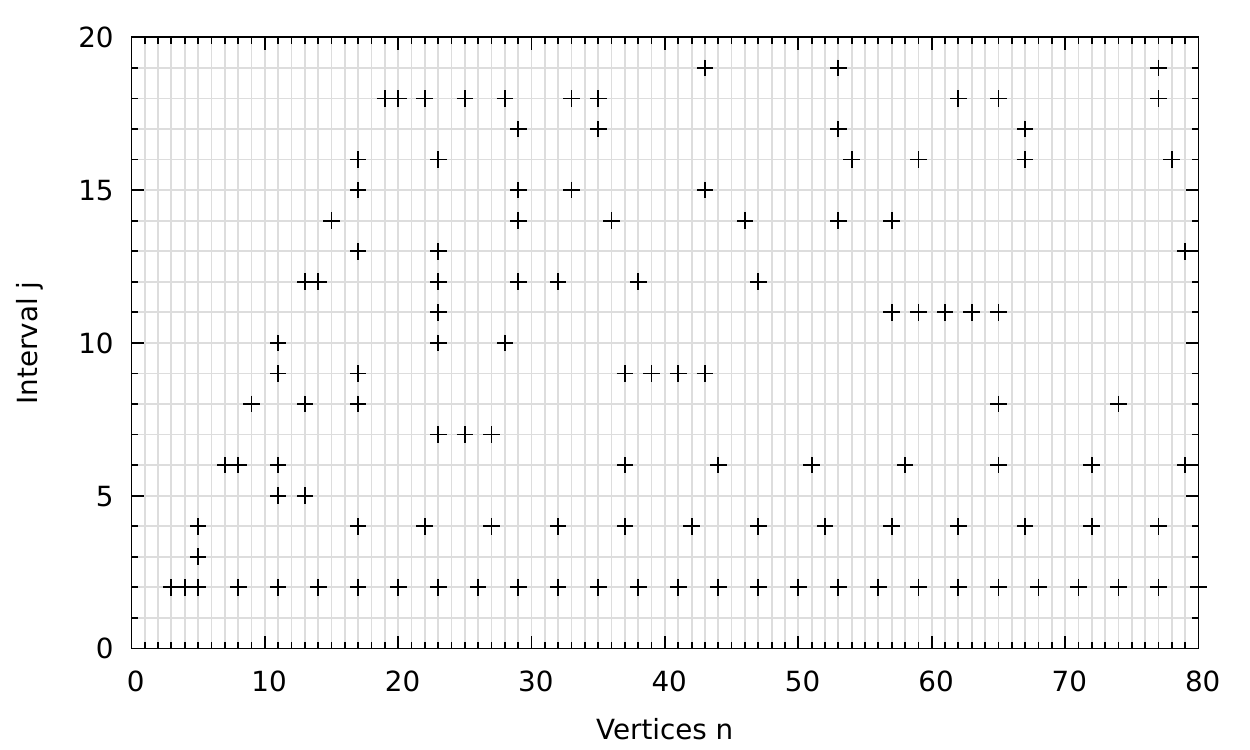}
    }
  \end{center}
  \vspace*{-1.0em}
  \caption{Visualization of critical circulant graphs of degree $4$.}
  \label{fig:critical-circulant-graphs-2}
\end{figure}

To conclude this section we will present a general rule 
for a subset of circulant graphs to determine whether they are
critical or not. One can see that Fact~\ref{f:CycClique} 
is a conclusion of the following result.

\begin{definition}
  For $n, d_h \in \Nat$, define the undirected graph $C_{n,d_h} = (V_{n, d_h},
  E_{n, d_h})$ by
  $V_{n,d_h} = \{u_0,\ldots,u_{n-1}\}$ and 
  $E_{n,d_h} = \ \bigcup_{j=0}^{n-1} \{\{u_j,v\} | v\in N^+_{j,n,d_h} 
                  \cup N^-_{j,n,d_h}\}$,
  where for all $i \in \{0,\ldots, n-1\}$, it holds
  \begin{align*}
    N^+_{i,n,d_h} &=  \{ u_{(i+1)\bmod n}, .., u_{(i+d_h)\bmod n}\}~\text{and} \\
    N^-_{i,n,d_h} &=  \{ u_{(i-1)\bmod n}, .., u_{(i-d_h)\bmod n}\}.
  \end{align*}
For $n, d_h, \delta,i \in \Nat$ and the graph $C_{n,d_h}=(V_{n,d_h},E_{n,d_h})$ 
define the sets of vertices
\begin{eqnarray*}
V_{i,n,d_h}^\delta & = & \{u_{(i\cdot (d_h+1)+\delta)\bmod n}, \ldots, \\
& & \qquad u_{((i+1)\cdot (d_h+1)+\delta-1)\bmod n}\}\ .
\end{eqnarray*}
\end{definition}

Thus, $C_{n,d_h}=\CirculantGraph(n, L)$ where $L=[1,d_h]$.

\begin{figure}[h]
  \begin{center}
    \scalebox{.25}{
      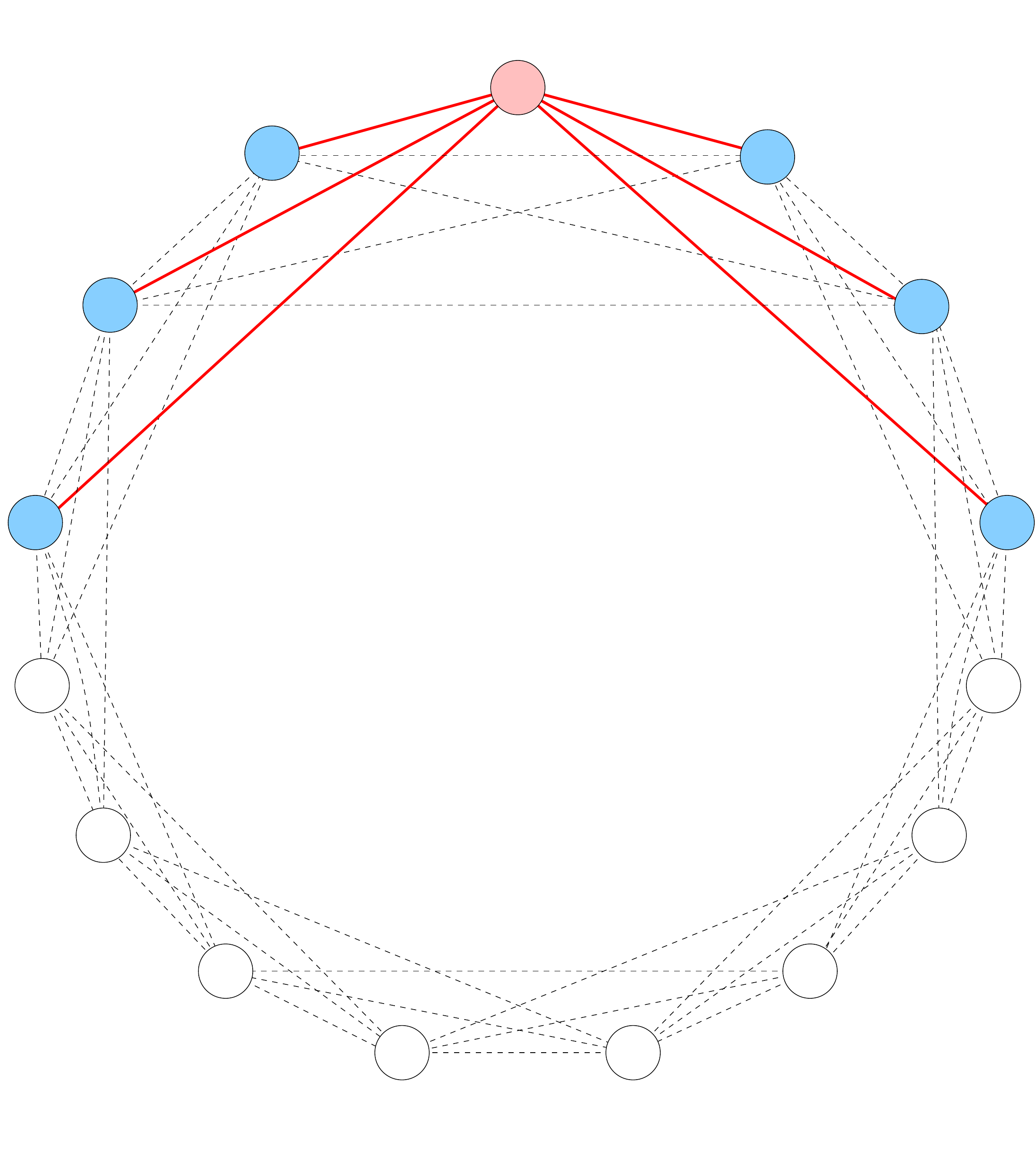
    }
  \end{center}
  \vspace{-1em}
  \caption{The \Irreducible{} graph $C_{15, 3}$.}
  \label{fig:C15-3}
\end{figure}

\begin{lemma}
  \label{lem:C_nd-MVC}
  The minimum vertex cover of each graph $C_{n, d_h}$ is $n - \left \lceil
  \frac{n - d_h}{d_h + 1} \right\rceil$.
\end{lemma}
\begin{proof}
One can easily verify that each subset $V_{i,n,d_h}^\delta$ forms a complete graph. Thus, 
every vertex cover of a graph $C_{n,d_h}$ has to include all vertices of every 
subset $V_{i,n,d_h}^\delta$ except of at most one vertex. W.l.o.g., we can assume that $u_0$ 
is not an element of the cover.
Thus, there exists at most one vertex that is not an element of the cover for each
subgraph $V_{i,n,d_h}^0$ for $i\in\{0,\ldots,\lceil (n-d_h)/(d_h+1)\rceil\}$. Since 
also all vertices of $N^-_{0,n,d_h}$ have to be in this cover.

A cover of size $n-\lceil (n-d_h)/(d_h+1)\rceil$ is given by $C=V\setminus \{\ u_{i\cdot(d_h+1)}\ |\ 
i\in\{0,\ldots,\lceil (n-d_h)/(d_h+1)\rceil\}\ \}$.
\end{proof}


\begin{theorem}
  \label{th:C_nd-MVC}
  $C_{n, d_h}$ is a connected \Irreducible{} graph iff either $n \le 2 d_h + 1$
  or $n - d_h$ is a multiple of $d_h + 1$.
\end{theorem}
\begin{proof}
If $n\le 2d_h+1$ then $C_{n,d_h}$ is a clique and the claim follows directly. Thus it remains to 
show the claim for parameters $n$ and $d_h$ where $n > 2d_h+1$.

For $k\in\{1,\ldots,d_h\}$ 
let $C_{n,d_h}^k=(V_{n,d_h},E_{n,d_h}\setminus\{\{u_0,u_{n-k}\}\})$. 
Because of symmetry $C_{n,d_h}$ is \Irreducible{} iff for any $k\in\{1,\ldots,d_h\}$
the size of a minimum vertex cover has to be smaller than $n-\lceil (n-d_h)/(d_h+1)\rceil$.

Let us assume that for given parameters $n$ and $d_h$ with $n > 2d_h+1$, 
$C_{n,d_h}$ is \Irreducible{}. In the following, we will show that this implies that
$n-d_h$ is a multiple of $d_h+1$.

Let $C$ be a minimum vertex cover for $C_{n,d_h}^k$ and let $C=V_{n,d_h}\setminus\{u_{i_0},\ldots,u_{i_{\ell-1}}\}$
with $i_j<i_{j+1}$ for all $j\in\{0,\ldots,\ell-2\}$.
Given $n,d_h,$ and $k$, we will start to show that it is sufficient 
to focus on values of $n$ with $2d_h+1<n<3d_h+3$ and some specific values for $i_j$.

By the definition of $C_{n,d_h}^k$, we can conclude that $i_{j+1}-i_j\ge d_h+1$ for all $0\le j\le\ell-2$. 
Assume that for some $j$ with $0\le j\le\ell-3$, it holds that $i_{j+1}-i_j> d_h+1$ then we can decrease 
$i_{j+1}$ such that $i_{j+1}-i_j = d_h+1$ without decreasing $C$. Thus, we can assume in the following that 
$i_{j+1}-i_j = d_h+1$ for all $j$ with $0\le j\le\ell-3$. Furthermore, we can assume that $i_0\le 2d_0+1$, otherwise we can 
remove $u_{d_h}$ from $C$ to get a vertex cover of smaller size for $C_{n,d_h}^k$ -- thus, $C$ is not a minimum vertex cover 
for $C_{n,d_h}^k$. Analogously, it holds that $i_{\ell-1}-i_{\ell-2}\le 2d_h+1$ and $n-1-i_{\ell-1}\le 2d_h+1$.

Now, assume that $i_0>0$ or $i_{\ell-1}\not=n-k$ then $C$ is also a vertex cover for $C_{n,d_h}$ -- contradicting our 
assumption that the size of the minimum vertex cover of $C_{n,d_h}^k$ is smaller than the 
size of the minimum vertex cover of $C_{n,d_h}$ and $C$ is such a minimum vertex cover of $C_{n,d_h}^k$.
Thus, we assume that $i_0=0$ and $i_{\ell-1}=n-k$ in the following.

Finally if $(i_{\ell-1}-i_{\ell-2})+k-1\ge 2d_h+1$ then $(C\cup\{u_{i_{\ell-1}}\})\setminus \{u_{i_{\ell-2}+d_h+1}\}$
to get a vertex cover for $C_{n,d_h}^k$ and $C_{n,d_h}$ of size $\ell$ -- contradicting our assumption that
$C_{n,d_h}$ has a minimum vertex cover of size larger than the size of the minimum vertex cover of $C_{n,d_h}^k$.

Because the regularity of the graph as well as the regularity of the discussed
cover, we can assume that $\ell=3$ and therefore
\begin{equation}
\begin{array}[c]{rcl}
n & = & 1+d_h+1+(i_2-i_1-1)+1+(k-1) \\ 
 & = & d_h+k+(i_2-i_1)+1 \ \ < \ \ 3d_h+3
\end{array}
\label{eq:pr-th-C_nd-MVC-BoundForN}
\end{equation}
$$d_h \ \ \le \ \ i_2-i_1-1 \ \ < \ \ 2d_h+2\ .$$ 
Assume that for given parameters $n$ and $d_h$ the graph $C_{n,d_h}$ is
\Irreducible{}, then the minimum vertex cover size of $C_{n,d_h}^k$ for any
value $k\in\{1,\ldots,d_h\}$ has to be smaller than the size of the minimum
vertex cover of $C_{n,d_h}$. As shown above, it is sufficient to restrict our
analysis to $\ell=3$ with $i_0=0$ and $i_2=n-k$ and therefore $i_1\le n-d_h-1-k$
and $i_0\le n-2d_h-2-k$. If we consider the case where $k=d_h$, then this
implies that
$$
3d_h+2 \ \ \le \ \ n\ .
$$
Using Equation~\ref{eq:pr-th-C_nd-MVC-BoundForN} we get $n=3d_h+2$.
With respect to our simplifications above this implies:
\begin{quote}
If $C_{n,d_h}$ is \Irreducible{}, then $n-d_h$ is a multiple of $d_h+1$.
\end{quote}
It remains to show that $C_{n,d_h}$ is \Irreducible{}, if $n-d_h$ is a multiple of $d_h+1$.
Using our simplifications above it is sufficient to show:
\begin{quote}
The minimum vertex cover size of $C_{3d_h+2,d_h}^k$ for any value $k\in\{1,\ldots,d_h\}$
is $3d_h-1$. 
\end{quote}
More precisely, we will show that $\{u_0,u_{d_h+1},u_{3d_h+2-k}\}$ is an
independent set and therefore its complement is a vertex cover of the desired
size.

From our construction of $C_{3d_h+2,d_h}$ it follows that $u_{d_h+1}$ is not a neighbor 
of $u_0$. Furthermore, $C_{3d_h+2,d_h}^k$ is constructed by deleting the edge 
$\{u_0,u_{3d_h+2-k}\}$. Thus, it remains to show that $u_{d_h+1}$ and $u_{3d_h+2-k}$
are not directly connected. Note that
$$
(3d_h+2-k) - (d_h+1) \ = \ 2d_h+1-k \ > \ d_h\ .
$$
Hence $u_{3d_h+2-k}\not\in N^+_{d_h+1,3d_h+2,d_h}$. Hence, 
$C_{n,d_h}$ is \Irreducible{}, if $n-d_h$ is a multiple of $d_h+1$.
\end{proof}

\section{Extensions for the Efficient Generation of Graphs}
\label{sec:extensions}

The problem of deciding whether a given graph is $\alpha$-critical is complete
for the second level of the Boolean hierarchy $D_P=\cBH_2$ (see
\cite{papa:1988}, or \cite{joret:2007} for a survey on $\alpha$-critical
graphs), where $\cBH_2$ is the class of languages that form the intersection of
an $\cNP$ language with a co-$\cNP$ language. Thus, finding $\alpha$-critical
graphs is a hard problem. Moreover, for the generation of diverse critical
graphs, one needs efficient methods and patterns. In 1970, Wessel
\cite{wessel:1970} introduced a technique to merge two $\alpha$-critical graphs:

\begin{definition}
For a graph $G=(V,E)$ and a vertex $u\in V$, let $N_G(u)$ denote the
neighborhood of $u$ in $G$.

Given two connected graphs $G_1$, $G_2$, a distinguished edge $\{u,v\}\in
E(G_1)$, and a vertex $w \in V(G_2)$ with degree at least two. Let $G$ be the
disjoint union of $G_1$ and $G_2$. Then, for every neighbor $x\in N_{G_2}(w)$,
choose one vertex $y$ from $\{u,v\}$ and add the edge $\{x,y\}$ to $G$. Ensure
that $u$ and $v$ are each chosen at least once. Finally, remove the vertex $w$
and the edge $\{u,v\}$ from $G$. The resulting graph $G$ is said to be
\emph{pasted together} from $G_1$ and $G_2$.
\end{definition}

By investigating the maximum independent sets, Wessel showed the following
theorem:

\begin{theorem}{\bf [Wessel]}
If $G$ is pasted together from two connected $\alpha$-critical graphs $G_1$ and
$G_2$ of size at least three, then $G$ is $\alpha$-critical. If $G$ is a
connected $\alpha$-critical graph of size at least four and if $G$'s (vertex)
connectivity is two, then there exist two connected $\alpha$-critical graphs
$G_1$ and $G_2$ each of size at least three, such that $G$ is pasted together
from $G_1$ and $G_2$. The size of the maximum independent set of $G$ is given by
the sum of the sizes of the maximum independent sets of $G_1$ and $G_2$.
\end{theorem}

Following the proof of this obsevration (see e.g. \cite{wessel:1970,
lovasz:1993}), one can see that this method for extending $\alpha$-critical
graphs does not preserve the knowledge of a concrete maximum independent set or
a concrete minimum vertex cover in all cases. It preserves only the knowledge of
the corresponding cardinalities.

A second technique for enlarging an $\alpha$-critical graph can be found in
\cite{lovasz:1986}:

\begin{definition}\label{def:split}
Let $G$ be a $\alpha$-critical graph and $u$ be a vertex of $G$ with degree at
least two. The operation of splitting $u$ works as follows: Let $F$ denote a
strictly non-empty subset of the neighbors of $u$. Add two new vertices $v$ and
$w$ to $G$. Connect $v$ with $u$ and $w$. Connect all vertices of $F$ with $w$,
and delete the edges between $u$ and the vertices of $F$.

Let $\cG$ be a family of undirected graphs, then define $\splitExt{\cG}$ to be
the set of graphs that can be generated from any graph $G=(V,E)\in\cG$ by 
splitting any vertex of $G$. Let $\splitExtP{\cG}{*}$ denote the transitive closure 
of $\cG$
according to $\splitExtB$. If we apply the extension $k$ times, then the set of
resulting graphs is denoted by $\splitExtP{\cG}{k}$. If $\cG$ consists of a
single graph $G$, then we use the notions $\splitExt{G}$ and $\splitExtP{G}{*}$.
\end{definition} 

For the proof of the following theorem see e.g. \cite{lovasz:1986}:

\begin{theorem}\label{th:split-01}
Let $G=(V,E)$ be a connected $\alpha$-critical graph of size at least three.
Then, every vertex has degree at least two, and splitting any vertex $v \in V$
results in a new $\alpha$-critical graph $G'$ where the maximum independent set
and the minimum vertex cover is increased by one. Furthermore, if G' be a
connected $\alpha$-critical graph of size at least four. and if $v \in V(G')$
has exactly two neighbors $u, w$ in $G'$ where $u$ and $w$ are not adjacent;
then identifying $u$ with $w$ and deleting $v$ results in another
$\alpha$-critical graph $G$ where the maximum independent set and the minimum
vertex cover is decreased by one.
\end{theorem}

One can see that this method for extending $\alpha$-critical graphs
preserves the knowledge of a concrete maximum independent set and of a 
concrete minimum vertex cover.

\begin{observation}
Let $G=(V,E)$ be a connected $\alpha$-critical graph of size at least three
and let $v \in V$ be an arbitrary vertex of $G$. Let $U$ be an maximum 
independent set and $\overline{U}=V\setminus U$ be the corresponding 
minimum vertex cover. Let $v$ and $w$ denote the two new vertices, generated 
by the splitting operation of $v$ (see Definition~\ref{def:split}) and let
$G'$ be the resulting graph. Then
$$
U' \ = \ \left\{\begin{array}[c]{ll}
U\cup\{w\} & \text{if }u\in U\\
U\cup\{v\} & \text{if }u\not\in U
\end{array}\right.
\quad\text{and}\quad
\overline{U}' \ = \ \left\{\begin{array}[c]{ll}
\overline{U}\cup\{w\} & \text{if }u\in \overline{U}\\
\overline{U}\cup\{v\} & \text{if }u\not\in \overline{U}
\end{array}\right.
$$
determine an increased independent set and the corresponding vertex cover and
are therefore optimal.
\end{observation}

This section introduces a new method which can efficiently produce
critical graphs from extending cycles of odd length, called the
\emph{parallel extension} in (see Figure~\ref{fig:5Ext}a)).
Furthermore, we will discuss a restricted version of the
splitting operation, called \emph{chain
extension} (see Figure~\ref{fig:5Ext}b)) which had been used 
for our graph generator. 
Let us consider the
parallel extension first. Prior, we need a useful definition of what we call a
\VCO.

\begin{figure*}[h]
  \begin{center}
    \scalebox{.35}{
      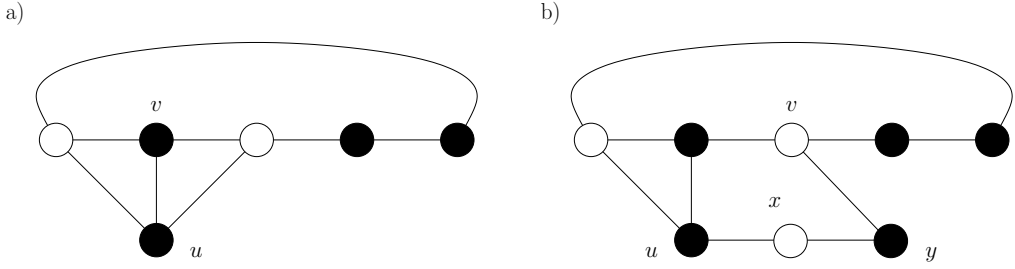
    }
  \end{center}
  \vspace{-1em}
  \caption{Two \Irreducible graphs which are generated from a cycle of five
  vertices.}
  \label{fig:5Ext}
\end{figure*}

\begin{definition}
  \label{def:overlap}
  Let $G = (V,E)$ be an undirected connected graph and let $U \subseteq V$ be a
  subset of vertices of $G$. $U$ is called a \VCO{} iff for every minimum
  vertex cover $C$ of $G$, it holds that $U \not\subseteq C$. If additionally,
  for any vertex $u \in U$, there exists a minimum vertex cover $C_u$ of $G$ such
  that $u \not\in C_u$ and $U \setminus \{u\}\subseteq C$, then $U$ is a \VCOO.
  Let $u$ be a new vertex and $U \subseteq V$. Then, define the extension of $G$
  according to $u$ and $U$ as
  $$
    \Gamma(G, U, u) = (V \cup \{u\}, E \cup \{ \UEdge{u, v} | v \in U \}).
  $$
\end{definition}

\begin{observation}
  \label{obs:VCover01}
  Given a graph $G = (V, E)$, a \VCO $U \subseteq V$, and a new vertex
  $u$, then the size of a minimum vertex cover of $\Gamma(G, U, u)$ is given by
  the size of a minimum vertex cover of $G$ plus one.
\end{observation}
\begin{proof}
Let $C$ be a minimum vertex cover of $G$, then $C\cup\{u\}$ is a vertex cover of 
$\Gamma(G,U,u)$. Consider now a minimum vertex cover $C'$ of $\Gamma(G,U,u)$, then
either $u\in C'$ or $N_{\Gamma(G,U,u)}(u)\subseteq C'$. 
In the first case $|C'|$ has to be larger than the minimum vertex cover size of
$G$ since $C'\setminus\{u\}$ is a vertex cover of $G$.
For the second case one can conclude that $|C'|$ is larger than the minimum
vertex cover size of $G$ since no minimum vertex cover of $G$ includes all
vertices of $U=N_{\Gamma(G,U,u)}(u)$.
\end{proof}


\begin{observation}
  \label{obs:VCover02}
  Given an \VCI graph $G = (V, E)$, a \VCO $U \subseteq V$, and a new vertex
  $u$. If $U$ is not a \VCOO, then $\Gamma(G, U, u)$ is \VCR.
\end{observation}
\begin{proof}
Since $U$ is a \VCO{} but not an \VCOO{} there exists a vertex $v\in U$ 
such that for every minimum vertex cover $C$ of $G$ with $v\not\in C$ there exists
a second vertex $v'\in U$ that is also not an element of $C$. Thus, $U\setminus \{v\}$
is a \VCO{} and the claim follows directly from Observation~\ref{obs:VCover01}.
\end{proof}


\begin{theorem}
  \label{th:VCover01}
  Given a graph $G = (V, E)$, a set $U \subseteq V$ with a \VCOO, and a
  new vertex $u$, then if $G$ is \VCI, then $\Gamma(G, U, u)$ is \VCI.
\end{theorem}
\begin{proof}
For the contrary, assume that $\Gamma(G,U,u)$ is \VCR, then we can either delete
an edge connected to $u$ or an edge of $E$ without changing the minimum vertex
cover size of $\Gamma(G,U,u)$:
\begin{enumerate}
\item Assume that we can delete the edge $\{u,v\}$ with $v\in U$ from
$\Gamma(G,U,u)$ without reducing the minimum vertex cover size of the resulting
graph $G'$. Since $U$ is a \VCOO, there exists a minimum vertex cover for $G$
that contains all nodes of $U\setminus\{v\}$. This vertex cover is also a vertex
cover of $G'$. Hence, deleting $\{u,v\}$ from $\Gamma(G,U,u)$ reduces the
minimum vertex cover size of the resulting graph.
\item Assume that we can delete an edge $e$ of $G$ without reducing the minimum vertex cover size of the resulting subgraph $G'$ of $\Gamma(G,U,u)$. Since $G$ is \VCI, there exists a 
vertex cover $C$ of $G$ after deleting $e$ where $|C|$ is smaller than minimum vertex cover size 
of $G$. Since $C\cup\{u\}$ is a vertex cover of $G'$, deleting the edge $e$ from $\Gamma(G,U,u)$
reduces minimum vertex cover size of the resulting graph.
\end{enumerate}
Thus, $\Gamma(G,U,u)$ is \VCI.
\end{proof}


Note that, if $U$ is given by a vertex of $G$ plus its neighborhood, then $U$ is
a \VCO. 
However, Theorem~\ref{th:VCover01} does not provide us with an efficient method
for constructing \Irreducible graphs. For this purpose, the following definition
simplifies the extension $\Gamma(G, U, u)$. 

\begin{definition}[\bfseries Parallel Extension]
  \label{def:par-ext}
  Let $G = (V,E)$ be an undirected connected graph and let $v \in V$ be a vertex
  of $G$. Then, define $G \setminus v = (V \setminus \{v\}, E \setminus
  \{\UEdge{v, u} | u \in V\})$ to be the graph $G$ without $v$. For
  an undirected graph $G = (V, E)$ and a vertex $v \in V$, let $N_G(v) = \{u \in
  V | \{v, u\} \in E\}$ denote the neighborhood of $v$ in $G$.

  Let $e = \{u,v\}\in E$ be an edge of $G$. Then, we call $u$ and $v$ \NeEq{} iff
  $N_G(v)\setminus\{u\}=N_G(u)\setminus\{v\}$. Moreover, let $\cG$ be a family
  of undirected graphs, then define the set of parallel extensions
  $\parExt{\cG}$ of $\cG$ to be the set of graphs that can be generated from any
  graph $G=(V,E)\in\cG$ by adding a new node $u\not\in V$ and new edges $E'$
  such that for at least one vertex $v\in V$, it holds that
  $E'=\{\UEdge{u,v'}|v'\in N_G(v)\cup\{v\}\}$.  
  We call $u$ the \emph{parallel extension} of $v$.

  Define $\parExtP{\cG}{0}=\cG$ and $\parExtP{\cG}{k} =
  \parExtP{\parExt{\cG}}{k-1}$. We write 
  \begin{align*}
    \parExtP{\cG}{*} = \bigcup_{k \in \{0, \ldots, \infty\}} \parExtP{\cG}{k}
  \end{align*}
  for the transitive closure of $\cG$. If $\cG = \{G\}$ consists of a single
  $G$, we usually write $\parExt{G}$ and $\parExtP{G}{*}$ instead of
  $\parExt{\cG}$ and $\parExtP{\cG}{*}$.
\end{definition}

A minimum vertex cover for any graph $G\in \parExt{G'}$
can easily be determined from a minimum vertex cover of $G'$.

\begin{theorem}
  \label{th:equivSubGraphCover}
  Let $G=(V,E)$ be an undirected connected graph and let $e = \UEdge{u,v} \in E$
  such that $u$ and $v$ are \NeEq{}. Let $U$ be a minimum vertex
  cover for $G\setminus u$. Then, $U \cup \{u\}$ is a minimum vertex cover for
  $G$.
\end{theorem}
\begin{proof}
Let $G=(V,E)$, $u,v\in V$, and $U$ fulfill the preconditions of the 
theorem. If $U$ is a cover of size $m$ for $G\setminus u$ then 
$U\cup\{u\}$ will be a cover
for $G$ of size $m+1$. Thus, it remains to show that if $U$ is a 
minimum cover for $G\setminus u$, then $U\cup\{u\}$ will be a minimum 
cover for $G$.

Assume that $U\cup\{u\}$ is not a minimum vertex cover for $G$. 
Then there exists a vertex cover $U'$ for $G$ with $|U'|\le m$. 
Note that for every vertex
cover for $G$ either $u$ or $v$ has to be in the cover.

Assume that $u\not\in U'$, then also the set 
$\{u\}\cup U'\setminus\{v\}$
has to be a vertex cover for $G$ of the same size.

But if $u$ is an element of a vertex cover $U'$ for $G$, 
then $U'\setminus \{u\}$ will be a vertex cover for 
$G\setminus u$ of size $m-1$ -- contradicting
the condition that $U$ is a minimum vertex cover for 
$G\setminus u$ and $U\cup\{u\}$ is a minimum
vertex cover for $G$. 
\end{proof}

\begin{corollary}
  \label{co:irredGraphMVCsize}
  Let $G$ be a \Irreducible{} graph with minimum vertex cover size $m$. Then,
  the minimum vertex cover size of each graph of $\parExt{G}$ is $m + 1$. More
  general, the minimum vertex cover size of each graph in $\parExtP{G}{k}$ is
  $m + k$.
\end{corollary}

We show that a parallel extension of any \Irreducible graph gives a new
\Irreducible graph.

\begin{theorem}
  \label{the:irred-par-irred}
  If $G$ is a connected \Irreducible{} graph, then all graphs in $\parExt{G}$
  are connected and \Irreducible.
\end{theorem}
\begin{proof}
The claim that any graph in $\parExt{G}$ is connected if 
$G$ is connected is obvious. Thus, we will focus on the claim
that any graph in $\parExt{G}$ is \Irreducible{} if 
$G$ is \Irreducible.

Let $G=(V_G,E_G)$ be a connected \Irreducible{} graph
with a minimum vertex cover of size $m$.
Assume that there exists a \Reducible{} graph 
$H=(V_H,E_H)\in\parExt{G}$. Furthermore, let
$v\in V_G$ and $u\in V_H$ such that $v,u$ 
defines the extension of $G$, i.e.
\begin{itemize}
  \item $v$ and $u$ are \NeEq{} for $H$,
  \item $u\not\in V_G$ and $G=H\setminus u$.
\end{itemize}
According to Corollary~\ref{co:irredGraphMVCsize} 
the minimum vertex 
cover size of $H$ is $m+1$. Since $H$ is \Reducible{}, some edge 
$e\in E_H$ exists, such that the minimum vertex cover size
of the smaller graph 
$H'=(V_H,E_H\setminus\{e\})$ 
is $m+1$, too. Let $U'$ be a minimum
vertex cover for $H'$. 

We distinguish three cases:
\begin{enumerate}
\item The edge $e$ is neither incident with $u$, nor with $v$. 
  In this case,
  $U'$ must contain $u$ or $v$ (or both).
  If $u \not\in U'$, then $v \in U'$, 
  and $\{u\} \cup U' \setminus \{v\}$ 
  would be another cover of $H'$ of the same size. 
  Thus, assume $u\in U'$.

  Since $u$ and $v$ are \NeEq{} for $H'$,
  Theorem~\ref{th:equivSubGraphCover} implies that 
  $U'\setminus \{u\}$ 
  determines a minimum vertex cover for 
  $G'=(V_G,E_G\setminus\{e\})$
  of size $m$. This is a contradiction to $G$ being \Irreducible. 
  
\item 
  The edge $e$ is incident with either $u$ or $v$ 
  (but not with both). By symmetry we can assume that $e$
  is incident with $v$.
  
  Since $G$ is irreducible, $G'$ has minimum vertex cover $U'$ of
  size $m-1$. On the other hand, all the  
  edges added to $G$ by the parallel extension are 
  incident with $u$; thus, $U'\setminus \{u\}$ determines a 
  minimum vertex cover for 
  $H'$ of size $m$ -- contradicting the assumption, that
  have the same 
  minimum vertex cover size. 

\item 
  The edge $e$ is incident with both $u$ and $v$, i.e., 
  $e=\UEdge{u,v}$. 
  
  Consider a minimum vertex cover $U$ 
  for $G$ with $|U|=m$.  
  If $v \not\in U$, then $U$ is a vertex cover of $H'$, 
  which would contradict the size $m+1$ for a 
  minimum vertex cover of $H'$. 
  Thus, assume $v \in U$. 
  
  Now, assume that a neighbor $w$ of $v$ is not in the cover $U$.
  By Observation~\ref{ops:2endpoints}, such a cover $U$ 
  exists, except 
  when $v$ is isolated. But then either $G$ is not connected 
  or $G$ consist of a single vertex.

  Now remove $\UEdge{w,v}$ from $G$ to get a smaller graph $G''$. 
  Since $G$ is \Irreducible{}, $G''$ has a vertex cover $U''$ 
  of size $m-1$. Note that $U''$ covers all edges in 
  in $H'$, except (possibly) $\UEdge{v, w}$ and the 
  connections between
  $u$ and the neighbors of $v$. If $v \in U''$, then 
  $U \cup \{u\}$ covers all edges in $H'$ (and even in $H$). 
  If $v \not \in U''$, all neighbors of $v$ 
  (except for $w$) would be in $U''$; thus, $U'' \cup \{w\}$ would 
  cover all edges in $H'$. 

  Since $|U'' \cup \{w\}| = m = |U'' \cup \{u\}|$, the size of our
  vertex cover for $H'$ contradicts the assumption of $H$ being 
  \Reducible.  
\end{enumerate}
\end{proof}


\begin{corollary}
  \label{co:irred-par-irred}
  Let $G$ be a \Irreducible{} graph. Then, all graphs of $\parExtP{G}{*}$ are
  \Irreducible.
\end{corollary}

We will now show an inverse observation to 
Corollary~\ref{co:irred-par-irred}. Analogously to Theorem~\ref{th:split-01}
our goal is to show that if an $\alpha$-critical graph $G$ is the parallel 
extension $G'$, then also $G'$ is $\alpha$-critical. 

\begin{theorem}
\label{th:irredSubGraphCover1}
Let $G=(V,E)$ be an \Irreducible{} connected 
graph and $e=\UEdge{u,v}\in E$ an
edge of $G$ such that $u$ and $v$ are \NeEq{}.
Then $H=G\setminus u$ is \Irreducible.
The size of a minimum vertex cover of $H$ is 
reduced by 1 according to the 
minimum vertex cover of $G$.
\end{theorem}

\begin{proof}
The second claim, that the size of a minimum vertex cover of 
$H=(V_H,E_H)$ is reduced by 1 according to the minimum vertex 
cover of $G$ follows directly from
Theorem~\ref{th:equivSubGraphCover} if $H$ is irreducible.

To show that $H$ is \Irreducible{} if $G$ is \Irreducible{}, 
we have transform the
minimum vertex cover of $G$ after deleting an edge 
$e'\in E\cap E_H$, which
occurs in both graphs, into a minimum vertex cover of $H$.

Assume that the size of a minimum vertex cover of $G$ is $m$. 
Since $G$ is \Irreducible{}, the size of the minimum vertex cover 
of $G'=(V,E\setminus\{e'\})$ is $m-1$. Since $e'\in E\cap E_H$ 
the edge $e'$ is not incident with $u$. Thus,
each vertex cover of $G'$ must contain either $u$ or $v$ or both. 
Since $N_{G'}(v)\subseteq N_{G'}(u)$, we can easily transform each 
vertex cover for $G'$ into a vertex cover for $G'$ that contains 
$u$ and has the same size. Hence,
such a minimum vertex cover for $G'$ gives us a minimum vertex 
cover for $G'\setminus u$ of size $m-2$. 
Since $G'\setminus u$ is isomorphic to the graph $H$
after deleting $e'$, $H$ is \Irreducible{}.
\end{proof}

Note that the observations on the size of a minimum 
vertex cover of the shrinked
graphs presented in
Theorem~\ref{th:irredSubGraphCover1} also work
for \Reducible{} graphs. Thus, we get:

\begin{observation}
\label{ob:shrinkSubGraphCover1}
Let $G=(V,E)$ be an undirected connected 
graph with minimum vertex cover size
$m$.
Let $e=\UEdge{u,v}\in E$ be an edge of $G$ 
such that $u$ and $v$ are \NeEq{}. 
Then size of a minimum vertex cover of 
$H=G\setminus u$ is $m-1$.
\end{observation}

Let $\cC$ denote the family of all cliques and let $\cOC$ denote the family of
all cycles of odd length. For easier notion, we assume that a graph that
consists of only a single node is in $\cC$ and in $\cOC$. Furthermore, we assume
that a graph which consists of only two connected nodes is in $\cC$. Then, we
have:

\begin{observation}
  \label{ob:start_with_3cyc}
  $\cC \subset \parExtP{\cOC}{*}$ and $\cC = \parExtP{G_1}{*}$ where 
  $G_1 = (\{v\}, \emptyset)$. Moreover, if 
  $G_2 = (\{u, v\}, \{\UEdge{u, v}\})$ and if $G_3$
  denotes the cycle of length three, then we have 
  $\cC \setminus \{G_1\} = \parExtP{G_2}{*}$ and 
  $\cC \setminus \{G_1,G_2\} = \parExtP{G_3}{*}$.
\end{observation}

Thus, simple edges and cycles of length three are interesting candidates to
start with a generation process for \Irreducible graphs. Based on the following
observation, one can show that the relation of neighbor-equivalence can be used
to define equivalence classes of vertices.

\begin{observation}
  \label{obs:closureEquiv1}
  Given an undirected connected graph $G = (V, E)$. Assume that $u$ is \NeEq{}
  to $v$ and to $w$, then $v$ is \NeEq{} to $w$.
\end{observation}
\begin{proof}
Since $u$ is equivalent to $v$ and to $w$ the two nodes 
$v$ and $w$ have to be
adjacent. Furthermore, by definition we have
$N_G(w)\setminus\{u,v\}=N_G(u)\setminus\{v,w\}=N_G(v)\setminus\{u,w\}$.
Hence, $N_G(w)\setminus\{v\}=N_G(v)\setminus\{w\}$ and therefore 
$v$ is equivalent to $w$.
\end{proof}

\begin{definition}
Let $G=(V,E)$ be an undirected connected graph and $u\in V$. 
By 
$\parCla{u}$
we denote the set of all vertices that are \NeEq{} to $u$ 
plus $u$ itself. 
We call the different sets $\parCla{u}$ of $G$ the node 
equivalence classes of $G$.
\end{definition}

\begin{observation}
\label{obs:equivClCyc}
Every cliques has only one node 
equivalence class. For $i\ge 4$ every cycle of 
length $i$ has exactly $i$
node equivalence classes.
\end{observation}

From Observation~\ref{obs:closureEquiv1} we conclude:

\begin{corollary}
\label{obs:closureEquiv2}
Let $G=(V,E)$ be an undirected connected graph and $u\in V$. 
Then for all $v\in \parCla{u}$
we have $\parCla{u}=\parCla{v}$.
\end{corollary}

\begin{theorem}
\label{th:GammaAndEquivalenceClasses}
Let $G$ be an undirected connected graph of $k$ node equivalence 
classes. Then every graph in $\parExtP{G}{*}$ has exactly
$k$ node equivalence classes.
\end{theorem}
\begin{proof}
Since every graph $G'\in\parExtP{G}{*}$ can be generated 
by sequence of steps
where in each step a new node is added (together with some 
edges connecting this
node to the rest of the graph) to the graph, it sufficient 
to prove that such a
step does not change the number equivalence classes. Thus, 
we have to analyze
such a step:

Let $G'=(V',E')$ be an arbitrary graph, $u$ be a new node and 
$V''\subseteq V'$ be a subset of nodes such that for at least 
one vertex $v\in V''$ it holds that $N_G(v)=V''\cup\{v\}$.
Then the new graph is given by
$G'''=(V'\cup\{u\}, E'\cup\{\UEdge{u,v'}|v'\in V''\})$.

From the construction it follows that $u$ and $v$ are equivalent 
in $G'''$; thus, it remains to show that for any pair of nodes 
$x,y$ of $G'$ it holds that
\begin{enumerate}
\item if $x$ and $y$ are not equivalent in $G'$ then they 
are not equivalent in $G'''$, and
\item if $x$ and $y$ are equivalent in $G'$ then they are 
equivalent in $G'''$.
\end{enumerate}
Since our step to extend $G'$ only adds new edges to the 
graph that are incident
with the new node, this step does not extend the 
neighborhood of an old node according to the node in $V'$. 
Thus, if two nodes are not equivalent in $G'$ then they are not
equivalent in $G'''$.

Now assume that $x$ and $y$ are equivalent in $G'$. 
If the new node is a parallel extension of $x$ then this extension 
adds the new node also
to the neighborhood of $y$. 
Thus, $x$ and $y$ remain equivalent in $G'''$.
\end{proof}

From Observation~\ref{obs:equivClCyc} and 
Theorem~\ref{th:GammaAndEquivalenceClasses}
we can directly conclude that:

\begin{corollary}
Let $G_i$ be a cycle of length $2i+1$ and let $k,m\ge 1$ with $k\not=m$ the
$\parExtP{G_k}{*}\cap \parExtP{G_m}{*}=\emptyset$.
\end{corollary}


The analysis of the equivalence classes leads to the observation that the sets
of graphs which can be generated by parallel extensions from two different
cycles of different odd length $i, j \ge 3$ are always distinct. Thus, we have to 
use for a second type of extension that helps us to generate cycles of odd length.
Consider the \Irreducible{} graph of Figure~\ref{fig:5Ext}.b) which cannot be
generated by parallel extensions of a cycle. This graph yields a second type
of extension, the splitting extension (see Definition~\ref{def:split}).


Investigating the vertex covers of cycles of odd length, one can see
that every minimum cover contains two adjacent vertices. Investigating 
bipartite graphs, one can see that this is not a general behaviour of
all graphs. In the following we will investigate whether this
behaviour will be preserved by parallel and splitting extensions:

\begin{definition}
  \label{def:double_cover_condition}
  An edge $e = \UEdge{u,v}$ of a graph $G$ fulfills the double-cover condition
  if there exists a minimum vertex cover $U$ of $G$ with $u,v\in U$. A graph $G$
  fulfills the double-cover condition if any edge of $G$ fulfills it.
\end{definition}

\begin{lemma}
\label{lem:chainExt_with_2endpoints_in_cover_02}
Let $G=(V,E)$ be an \Irreducible{} graph such that 
for every edge $e=\UEdge{u,v}$ there
exists a minimum vertex cover $U$, such 
that both endpoints of $e$ are in the
cover $U$. Then for every graph 
$G'=(V',E')\in \splitExt{G}\cup \parExt{G}$ and
for every edge $e'=\UEdge{u',v'}$ there exists a 
minimum vertex cover $U'$, such
that both endpoints of $e'$ are in the cover $U'$.
\end{lemma}
\begin{proof}
Assume that $G'\in \splitExt{G}$ and assume that 
$G'=(V',E')$ can be generated by splitting $u''\in V$, i.e.
there exists a subset $F\subset N_{G}(u'')$ and $x,y\in
V'\setminus V$ such that
\begin{eqnarray*}
V' & = & V\cup\{x,y\}\\
E' & = & \{\UEdge{u'',x},\UEdge{x,y}\} \cup \{\UEdge{y,v''}|v''\in F\}
\cup E\setminus\{\UEdge{u'',v''}|v''\in F\}\ .
\end{eqnarray*}
For every edge $e$ let $U_e$ denote a minimum vertex cover 
of $G$ that contains both endpoints of $e$.

Now assume that $e\not\in \{\UEdge{u'',v''}|v''\in F\}$. Since either $u''$ or all
vertices in $F$ have to be in $U_e$, we
can extend the cover $U_e$ to a vertex cover 
for $G'$ by adding either $x$ or
$y$ to $U_e$. The resulting set is a minimum 
vertex cover for $G'$ that covers
both endpoints of $e$.

Next assume that $e\in \{\UEdge{u'',v''}|v''\in F\}$. 
Since both nodes $u''$ and $v''$ are in 
$U_e$, we can
extend the cover $U_e$ to a vertex cover for $G'$ by adding either 
$x$ or $y$ to
$U_e$. Depending on the added node we get a minimum vertex 
cover for $G'$ that
covers either both endpoints of $\UEdge{u'',x}$ or of $\UEdge{y,v''}$
for all $v''\in F$.

To get a minimum vertex cover that contains $x$ and $y$, 
let us focus on a
minimum vertex cover $U$ for $G$ after deleting an edge 
$\UEdge{u'',v''}$ for $v''\in F$. Since $G$ is $\alpha$-critical,
the maximal independent set of the resulting graph will be increased 
and includes both vertices $u''$ and $v''$. This implies that
the minimum vertex cover $U$ will be reduced by 1 (compared to the
minimum vertex cover of $G$) and neither contains $u''$ nor $v''$. 
Thus, adding $x$ and $y$
to $U$ gives us a minimum vertex cover for $G'$ that covers 
both
endpoints of $\UEdge{x,y}$.

Let us now focus on graphs $G'\in \parExt{G}$. 
Assume that $G'=(V',E')$ and let
$v\in V$ and $u\in V'$ such that $u$ is the parallel
extension of $v$, i.e.
\begin{itemize}
  \item $v$ and $u$ are \NeEq{} for $G'$,
  \item $u \not\in V$ and $G = H\setminus u$.
\end{itemize}
Since every minimum vertex cover $U$ for $G$ can be transformed 
into a minimum vertex cover for $G'$ by adding $u$ to $U$ (see
Theorem~\ref{th:equivSubGraphCover}), 
the claim holds for all edges of $G'$ that
are already in $G$. Since $G$ fulfills the double cover
condition, there exists a minimum vertex cover $U$ for
each vertex $v'$ that contains $v'$. Since $U\cup\{u\}$ 
is a minimum vertex cover for $G'$, the new edge $\UEdge{u,v'}$
fulfills the double cover property.
\end{proof}

This lemma implies:

\begin{theorem}
  \label{th:chainIrred01}
  Let $\cG$ be a family of \Irreducible{} undirected graphs. If $\cG$ neither
  contains the graph with a single node nor the graph with a single edge and if
  each graph of $\cG$ fulfils the double-cover condition, then all graphs of
  $\splitExt{\cG}$ and $\parExt{\cG}$ are \Irreducible and fulfil the 
  double-cover condition.
\end{theorem}

\begin{corollary}
  \label{coro:chainIrred01}
  Let $\cG$ be a family of \Irreducible{} undirected graphs without the
  single-node graph and without the single-edge graph, then any graph of
  $\splitExtP{\cG}{*}$ is \Irreducible. If in addition all graphs of $\cG$
	fulfil the double-cover condition then all graphs of $\splitExtP{\cG}{*}$
	fulfil the double-cover condition.
\end{corollary}

To combine splitting and parallel extensions we define: 

\begin{definition}
  Let $\cG$ be a family of undirected graphs, then define 
  $$
    \splitParExt{\cG} \ \ = \ \ \splitExt{\cG}\cup\parExt{\cG}\ .
  $$ 
  Let $\splitParExtP{\cG}{*}$ denote the
  transitive closure of $\cG$ according to $\splitParExtB$. If we apply the
  extension $k$ times, then the set of resulting graphs is denoted by
  $\splitParExtP{\cG}{k}$. If $\cG$ consists of a single graph $G$, we use the
  notions $\splitParExt{G}$ and $\splitParExtP{G}{*}$.
\end{definition}

There are connected \Irreducible{} graphs that cannot be generated by splitting and
parallel extensions from a cycle of length three. An example is illustrated in
Figure~\ref{fig:IrreducibleNotPC2}.

\begin{figure}[h]
  \begin{center}
    \scalebox{.3}{
      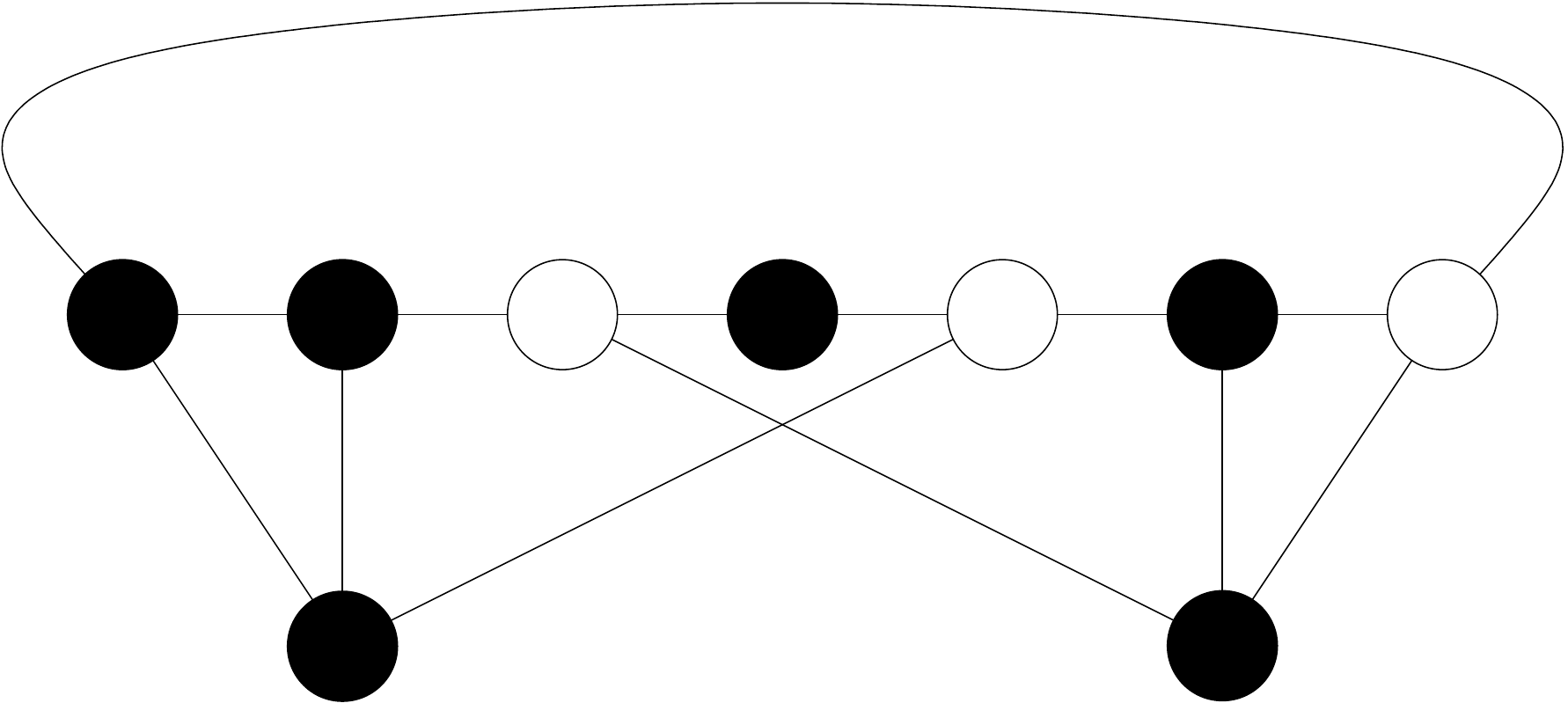
    }
  \end{center}
  \vspace{-1em}
  \caption{A \Irreducible{} graph that cannot be generated by splitting and
  parallel extensions from a cycle of length three.}
  \label{fig:IrreducibleNotPC2}
\end{figure}

For our experimental benchmark test we have restricted ourself a restricted 
version of our the splitting operations. More general experimental benchmark test
will be appear within a future work.

\begin{definition}[\bfseries Chain Extension]
  \label{def:chainExt01}
  Let $G = (V,E)$ be an undirected connected graph and let $e=\UEdge{u,v}\in E$
  be an edge of $G$. Let $x,y\not\in V$ denote two new nodes. Then
  the graph $G'=(V\cup\{x,y\},
  \{\UEdge{u,x},\UEdge{x,y},\UEdge{y,v}\}\cup E\setminus\{e\})$ 
  is called a \emph{chain extension} of $G$.

  Let $\cG$ be a family of undirected graphs, then define $\chaExt{\cG}$ to be
  the set of graphs that can be generated from any graph $G=(V,E)\in\cG$ by one
  chain extension. Let $\chaExtP{\cG}{*}$ denote the transitive closure of $\cG$
  according to $\chaExtB$. If we apply the extension $k$ times, then the set of
  resulting graphs is denoted by $\chaExtP{\cG}{k}$. If $\cG$ consists of a
  single graph $G$, then we use the notions $\chaExt{G}$ and $\chaExtP{G}{*}$.
	Analogously we define $\chaParExt{\cG}$, $\chaParExtP{\cG}{*}$, $\chaParExtP{\cG}{k}$,
	$\chaParExt{G}$, and $\chaParExtP{G}{*}$.
\end{definition}

Since the splitting operation is a generalization of 
the chain extension, can conclude the the chain extension
also preserves the $\alpha$-criticality of a graph and 
the knowledge of known minimum vertex covers and maximum 
independent sets. Furthermore, we can conclude that these 
properties are also preserved for the inverse 
chain extension, i.e. for the corresponding shrinkage 
operation. For the shrinkage of \Reducible{} we can 
show a similar result. This result follows directly from 
the following proof the properties of shrinkage 
of $\alpha$-critical graphs:

\begin{theorem}
\label{th:irredSubGraphCover2}
Let $G=(V,E)$ be an \Irreducible{} connected graph 
and let $x,y$ be two adjacent
nodes of $G$ of degree 2. Furthermore, let $u,v$ 
denote the two remaining
neighbors of $x$ and $y$ then the graph 
$H=(V\setminus\{x,y\}, \{\UEdge{u,v}\}\cup
E\setminus (\{x,y\}\times V))$ is \Irreducible{}. 
The size of a minimum vertex
cover of $H$ is reduced by 1 according to the 
minimum vertex cover of $G$.
\end{theorem}

\begin{proof}
Let us start to show the second claim. Let $m$ denote the size 
of a minimum vertex cover of $G$. By contradiction assume that 
there exists minimum vertex cover $U$ for $H=(V_H,E_H)$ of size 
$m-2$. Since $u$ and $v$ are adjacent, we can transform this vertex 
cover of $H$ into a vertex cover for $G$ by either
adding $x$ or $y$ to $U$. The resulting vertex cover has size $m-1$ --
contradicting our assumption the exists minimum vertex cover 
size of $G$ is $m$. Thus, the minimum vertex cover size is at least
$m-1$. On the other hand every vertex cover of $G$ can easily be 
transformed into a vertex cover for $H$, where the size of the
cover is reduced by one.

To show that $H$ is \Irreducible{} if $G$ is \Irreducible{}, 
we have to consider two cases:
\begin{enumerate}
  \item the deleted edge occurs in both graphs, 
          i.e. $e'\in E\cap E_H$,
  \item the deleted edge $e'\in E$ only occurs in $E$. 
\end{enumerate}
In both cases we transform the minimum vertex cover of $G$ 
after deleting $e'\in E$ into a minimum vertex cover of $H$ 
after deleting either $e'$ or, if
$e'\not\in E_H$, after deleting $\UEdge{u,v}$.

W.l.o.g.,\ assume that $u$ and $x$ are adjacent.

We start with the first case: 
Since $G$ is \Irreducible{}, the size of the minimum
vertex cover of $G'=(V,E\setminus\{e'\})$ is $m-1$. 
Since $e'\in E\cap E_H$ the
edge $e'$ is not incident with $x$ or $y$. Thus, each 
vertex cover of $G'$ must
contain at least two of the vertices $u,v,x,y$, and at least one 
of the vertices $x,y$. Thus, every the minimum vertex cover 
$U$ of $G'$ can be transformed into a
minimum vertex cover of $H'=(V_H,E_H\setminus\{e'\})$ such that
\begin{itemize}
  \item if the cover of $G'$ contains either $u$ or $v$, 
  then a cover of $H'$ is
  given by $U\setminus\{x,y\}$,
  \item if the cover of $G'$ contains neither $u$ nor 
  $v$, then it has to
  contain $x$ and $y$ and a cover of $H'$ is given by $\{u\}\cup
  U\setminus\{x,y\}$.
\end{itemize}
In both cases the size of the constructed cover for $H'$ is $m-2$.

Now, let us consider the case, that the deleted edge $e'$ only 
occurs in $E$, but not in $E_H$. For this case it suffice, if we 
restrict our selves to the case where $e'=\{u,x\}$. As in the 
previous case the size of the minimum vertex cover of 
$G'=(V,E\setminus\{e'\})$ is $m-1$. Furthermore, each minimum vertex
cover of $G'$ either contains $x$ or $y$. Thus, this vertex cover 
gives us directly a vertex cover for 
$H'=(V_H,E_H\setminus\{\UEdge{u,v}\})$ of size $m-2$.

Summarizing, our two cases implies that the vertex covers of $G'$ 
of size $m-1$ can be transformed into a vertex cover of $H'$ of size 
$m-2$. Since our analysis takes every edge of $H$ into account, $H$ 
is \Irreducible{}.
\end{proof}

Note that the observations on the size of a minimum 
vertex cover of the shrunk
graphs in
Theorem~\ref{th:irredSubGraphCover2} also work
for \Reducible{} graphs. Thus, we get:

\begin{observation}
\label{ob:shrinkSubGraphCover2}
Let $G=(V,E)$ be an undirected connected 
graph with minimum vertex cover size
$m$.
Let $x,y$ be two adjacent nodes of $G$ of degree 2 and let 
$u,v$ denote the two remaining neighbors of $x$ and $y$. Then the 
size of a minimum vertex cover of the
graph $H'=(V\setminus\{x,y\}, 
\{\UEdge{u,v}\}\cup E\setminus (\{x,y\}\times V))$ is
$m-1$.
\end{observation}

For generating our benchmark graphs, we focus on the set of base graphs $\cB$
that can be generated by parallel and chain extensions from a cycle of length
three. To efficiently generate such graphs, we used several bounds on the number
of edges according to the number of vertices and the size of a minimum vertex
cover. For a given graph $G = (V, E)$, we define $n(G) = |V|$ for the number of
vertices of $G$, $m(G) = |E|$ for the number of edges of $G$, $c(G)$ for the
size of the minimum vertex cover of $G$, and $\overline{c}(G) = n(G)-c(G)$ for
the number of vertices not in the minimum cover. 

In the following we will focus on some bounds for the number 
of edges $m(G)$ of a graph $G$, if the number of vertices 
$n(G)$ and the minimum 
cover size $c(G)$ are given.
We will first establish some general bounds for $m(G)$ according 
to $n(G)$ and
$c(G)$.
Since there are no edges
between two vertices that are not in the cover, we can conclude:

\begin{observation}
  \label{obs:bound01} For every graph $G$, it holds that 
\begin{eqnarray*}
m(G) & \le & \frac{c(G)\cdot (c(G)-1)}{2} + c(G)\cdot (n(G)-c(G)) \\ 
     & = & c(G)\cdot 
\left(n(G)-\frac{c(G)+1}{2}\right)\ .
\end{eqnarray*}
\end{observation}

Now we would like to establish a lower bound for the number of 
edges $m(G)$ for a graph if the number $n(G)$ of vertices of $G$ 
and the size of the minimum vertices cover $c(G)$ is given.

For a clique $K_i$ of $i$ vertices let 
$k_i=\frac{i\cdot (i-1)}{2}$ denote the
number of edges of $K_i$.

Given two vectors $\vec{u} = (u_1, \ldots, u_t)$ and $\vec{v} = (v_1, \ldots,
v_t)$, then we call $\vec{u}$ lexicographically smaller than $\vec{v}$, denoted
by $\vec{u} \lelex \vec{v}$, if either both vectors are equal or if there exists
an index $i \in \{1, \ldots, t\}$ such that $u_i < v_i$ and $u_j = v_j$ for all
$j \in \{1, \ldots, i - 1\}$.

\begin{definition}
Given two values $n$ and $c$ with $n>c$ define $\CoV(n,c)$ denote 
the set of vector 
$\vec{\alpha}=(\alpha_{n},\alpha_{n-1},\ldots,\alpha_{1})$ such that
$$
n = \sum_{i=1}^{n} i\cdot \alpha_i
\qquad\text{and}\qquad
c = \sum_{i=1}^{n} (i-1)\cdot \alpha_i\ .
$$
\end{definition}

Let us analyze the vectors
$\vec{\alpha}\in \CoV(n,c)$. 
We can easily see that:
\begin{itemize}
\item $n-c=\sum_{i=1}^{n}\alpha_i$; thus, the sum of all entries of 
  the vectors in $\CoV(n,c)$ is determined by the parameters $n$ and 
  $c$.
\item We will use the following interpretation of the vectors 
  $\vec{\alpha}=(\alpha_{n(G)},\alpha_{n(G)-1},\ldots,\alpha_{1})\in \CoV(n,c)$:
  To generate a graph of $n=n(G)$ vertices and of 
  minimum vertex cover size $c=c(G)$ we can use a collection of 
  cliques, where for
  $\alpha_i$
  determines the number cliques of size $i$.
\item If $\vec{\alpha}=(\alpha_{n},\alpha_{n-1},\ldots,\alpha_{1})$ 
  is a solution for
  $$
  n = \sum_{i=1}^{n} i\cdot \alpha_i
  \qquad\text{and}\qquad
  c = \sum_{i=1}^{n} (i-1)\cdot \alpha_i\ .
  $$
  and if $\alpha_i>0$ and $\alpha_j>0$ ($i\not=j+1$) then also the 
  vector
  $\vec{\alpha}_{i,j}=(\alpha_{n}',\alpha_{n-1}',\ldots,\alpha_{1}')$ 
  with
  $$\alpha_{k}'\ \ = \ \ \left\{\begin{array}[c]{ll}
  \alpha_{k} & \text{for }k\not\in\{i,i-1,j+1,j\}\\
  \alpha_{k}-1 & \text{for }k\in\{i,j\}\\
  \alpha_{k}+1 & \text{for }k\in\{i-1,j+1\}
  \end{array}\right.$$
  is a solution for the two equations.
\item If the vector $\vec{\alpha}_{i,j}$ is generated form 
  $\vec{\alpha}$ as described above and if $i>j+1$, then
  \begin{eqnarray*}
   & & \sum_{k=1}^{n} \frac{k\cdot (k-1)}{2}\cdot \alpha_k'\\
   & = &
  -\frac{i\cdot (i-1)}{2}+\frac{(i-1)\cdot (i-2)}{2}+\frac{(j+1)\cdot j}{2}\\
	& & \qquad -\frac{j\cdot (j-1)}{2} + \sum_{k=1}^{n} \frac{k\cdot (k-1)}{2}\cdot \alpha_k\\
  & = & -\frac{(i-i+2)\cdot (i-1)}{2}+\frac{(j+1-j+1)\cdot j}{2}\\
	& & \qquad + \sum_{k=1}^{n} \frac{k\cdot (k-1)}{2}\cdot \alpha_k\\
  & = & -i+1+j+ \sum_{k=1}^{n} \frac{k\cdot (k-1)}{2}\cdot \alpha_k\\
  & < & \sum_{k=1}^{n} \frac{k\cdot (k-1)}{2}\cdot \alpha_k
  \end{eqnarray*}
  Thus, the corresponding value of the sum of edges in the collection 
  of cliques is smaller for $\vec{\alpha}_{i,j}$ than for 
  $\vec{\alpha}$ if $i>j+1$.
\item Hence, the lexicographical minimum vector results in a 
  collection of cliques with the minimum number of edges.
\item Moreover the lexicographical minimum vector includes at most 
  two positions $i,i+1$ where the entries are greater than $0$ (the 
  values on all other positions are 0).
\end{itemize}

\begin{lemma}{\bf [Lower-Bound Lemma]}
  \label{lem:bound01}
  Let $G = (V, E)$ be a graph generated according to
  Definition~\ref{def:random_processes}, where $\cB$ is given by the graphs in
  $\chaParExtP{K_3}{*} \cup \{K_1, K_2\}$. Let
  $(\alpha_{n(G)}, \alpha_{n(G)-1}, \ldots, \alpha_{1})\ \ \in\ \ \Nat^{n(G)}$
  denote the lexicographically minimal vector that fulfills the following two
  equations
  $$
    n(G) = \sum_{i=1}^{n(G)} i\cdot \alpha_i
    \qquad\text{and}\qquad
    c(G) = \sum_{i=1}^{n(G)} (i-1)\cdot \alpha_i\ .
  $$
  Then, we have
  $$
    m(G) \ \ \ge \ \ \sum_{i=1}^{n(G)} 
    \frac{i\cdot (i-1)}{2}\cdot \alpha_i\ .
  $$
\end{lemma}
\begin{proof}
Let $G=(V,E)$ be a graph generated by according to
Definition~\ref{def:random_processes} where 
$\cB$ is given by the graphs in
$\chaParExtP{K_3}{*}\cup\{K_1,K_2\}$.
Furthermore, assume that $G$ is not \Irreducible{}. 
Then there exists at least one edge that has be
added to $G$ by the process $\cG_{m}^3(\cdot)$, such
that removing this edge from $G$ does not reduce 
the minimum size of a cover and
the number of vertices of the resulting graph $G'$. 
Thus, the lower bound on the
number of edges of this claim for $G$ and $G'$ are equal. 
In the following we
will assume that $G$ is \Irreducible{} and 
generated by according to
Definition~\ref{def:random_processes} where 
$\cB$ is given by the graphs in
$\chaParExtP{K_3}{*}\cup\{K_1,K_2\}$.

Note that if the connected components of $G$ are cliques, 
then these components
determine a vector from $\CoV(n,c)$. 
Thus, the claim follows directly.

Now consider the case that not all the connected components of 
$G$ are cliques
and assume that for the lexicographically minimal vector
$\vec{\alpha}\in\CoV(n(G),c(G))$ we have
$$
m(G) \ \ < \ \ \sum_{i=1}^{n(G)} \frac{i\cdot (i-1)}{2}\cdot \alpha_i\ .
$$
By focusing on the connected components $G'$ of $G$, 
we can conclude that there has to exist 
such a component
$G'\in \chaParExtP{K_3}{*}$ which is not a clique and that 
for the lexicographically minimal vector 
$\vec{\alpha}'\in\CoV(n(G'),c(G'))$
we have
$$
m(G') \ \ < \ \ \sum_{i=1}^{n(G')} \frac{i\cdot (i-1)}{2}\cdot \alpha_i'\ .
$$
We will now show that such a graph $G'$ does not exist. 

For the contrary, assume that such a connected \Irreducible{}
graph exists in $\chaParExtP{K_3}{*}$
and that $G'$ denotes such a graph of minimum size.

Since $G'$ is not a clique and since there exists a sequence of 
graph extensions
$\chaExt{\cdot}$, $\parExt{\cdot}$ for generation $G'$ 
from $K_3$, this sequence contains at least one chain 
extensions $\chaExt{\cdot}$.
Let us fix
such a sequence $\Gamma_1,\ldots,\Gamma_k$ such that 
$G'=\Gamma_k(\cdots
\Gamma_1(K_3)\cdots)$. 

One can even show that this sequence contains 
exactly $n(G')-c(G')-1$
chain extensions $\chaExt{\cdot}$.

Let $G_i'=\Gamma_i(\cdots \Gamma_1(K_3)\cdots)$ denote the 
intermediate graphs
on the construction of $G'$. And assume that $\Gamma_\ell$ 
had been the last
chain extensions within the sequence $\Gamma_1,\ldots,\Gamma_k$, 
i.e. all
$\Gamma_h$ with $\ell<h\le k$ are parallel extensions. 

We will now compare $G'$
and $G''=\Gamma_k(\cdots \Gamma_{\ell+1}(G_{\ell-1}'+K_2)\cdots)$ where
$G_{\ell-1}'+K_2$ denotes the graph $G_{\ell-1}'$ plus an separate component
$K_2$ where the two nodes of $K_2$ will have the same name labels than the two
new nodes that are added to $G_{\ell-1}'$ by performing the chain extension
$\Gamma_\ell$. Since $\Gamma_h$ with $\ell<h\le k$ are parallel extensions we
can conclude that
\begin{eqnarray*}
n(G')-c(G') & = & n(G_\ell')-c(G_\ell') \\ 
& = & 
n(G_{\ell-1})-c(G_{\ell-1})-1 \\ 
& = & n(G_{\ell-1}'+K_2)-c(G_{\ell-1}'+K_2) \\ 
& = & n(G'')-c(G'')\ .
\end{eqnarray*}
Since $n(G')=n(G'')$ we have $c(G')=c(G'')$. Furthermore, since 
$\Gamma_h$ with $\ell<h\le k$ are parallel extensions we
can conclude that $G''$ consists of two connected components,
where at least the components that has been constructed form $K_2$
is a clique.

Let $G_i''=\Gamma_i(\cdots \Gamma_{\ell+1}(G_{\ell-1}'+K_2)\cdots)$ denote the
intermediate graphs on the construction of $G''$ after step $\ell$ (the addition
of $K_2$). By induction one can show that for every parallel extension
$\Gamma_{\ell+j}$ ($1\le j\le k-\ell$) the degree of the duplicated vertex of
$G_{\ell+j-1}''$ is smaller or equal than the degree of the corresponding
duplicated vertex of $G_{\ell+j-1}'$. Hence, $m(G'')\le m(G')$ and therefore
$$
m(G'') \ \ < \ \ m(G') \ \ < \ \ \sum_{i=1}^{n(G')} \frac{i\cdot (i-1)}{2}\cdot \alpha_i'\ .
$$
Again, by focusing on the connected components $G'''$ of $G''$, 
we can conclude that there has to exist 
such a $G'''\in \chaParExtP{K_3}{*}$ which is not a clique and that 
for the lexicographically minimal vector $\vec{\alpha}'''\in\CoV(n(G'''),c(G'''))$
we have
$$
m(G''') \ \ \le \ \ \sum_{i=1}^{n(G''')} \frac{i\cdot (i-1)}{2}\cdot \alpha_i'''\ .
$$
Since $G'''$ is a proper subgraph of $G''$ we can conclude that the size of 
$G'''$ is strictly smaller than size of $G'$ -- a contradiction to our choice
of $G'$.
\end{proof}


We conclude this section with the observation that the vectors $\vec{\alpha} =
(\alpha_{n},\ldots,\alpha_{1})$ fulfill
$$
  n = \sum_{i=1}^{n} i\cdot \alpha_i
   \qquad\text{and}\qquad 
  c = \sum_{i=1}^{n} (i-1)\cdot \alpha_i\ .
$$
One can show that if for $\vec{\alpha}=(\alpha_{n},\ldots,\alpha_{1})$ there are
two indices $i,j$ with $\alpha_i,\alpha_j>0$ and $i\not\in\{j,j+1\}$ then also
the vector $\vec{\alpha}_{i,j}=(\alpha_{n}',\ldots,\alpha_{1}')$ with
$$
  \alpha_{k}'\ \ = \ \ \left\{\begin{array}[c]{ll}
  \alpha_{k} & \text{for }k\not\in\{i,i-1,j+1,j\}\\
  \alpha_{k}-1 & \text{for }k\in\{i,j\}\\
  \alpha_{k}+1 & \text{for }k\in\{i-1,j+1\}
  \end{array}\right.
$$
is a solution for the two equations on $c$ and $n$. 
On the other hand, if 
$i>j+1$, then
$$
  \sum_{k=1}^{n} \frac{k\cdot (k-1)}{2}\cdot \alpha_k'
  \quad < \quad
  \sum_{k=1}^{n} \frac{k\cdot (k-1)}{2}\cdot \alpha_k\ .
$$ 
Thus, we can conclude that the lexicographical minimal vector $\vec{\alpha}$ as
at most two non-zero entries. And if it has two non-zero entries, then one can
find these entries at two consecutive positions $h, h+1$. Hence, the
lexicographical minimal vector can be determined as follows
\begin{quote}
  Let $h = \lceil c/(n - c) \rceil + 1$ and $g = \lfloor c/(n-c) \rfloor + 1$.
  If $h = g$ choose $\alpha_h = (n - c)$ and $\alpha_i=0$ for all $i\not=h$. If
  $h \not= g$ choose $\alpha_h = c \bmod (n - c)$, $\alpha_g = n - c -
  \alpha_h$, and $\alpha_i = 0$ for all $i \not\in \{h, g\}$.
\end{quote}
Note that a cycle of odd length can be generated by a sequence of chain
extensions of an $K_3$. 

For our graph generator we use a process that generates 
random graphs from $\chaParExtP{K_3}{*}$ for $\cG_{\cB,\ell,m,n}^1$.
Since it might be possible that this process generates 
huge components, such that these components can not be 
combined with earlier chosen components 
from $\chaParExtP{K_3}{*}$ it will be useful to 
see whether an parallel or an chain extension of an interim
graph leads to a component that cannot used within the 
further construction process. 

One idea for such a decision process will be to investigate 
the vectors $(\alpha_{n},\alpha_{n-1},\ldots,\alpha_{1})$ 
for the remaining unused vertices 
after each chain or parallel extension \chaExtB\ or \parExtB.
In the following we will investigate the changes of 
$(\alpha_{n},\alpha_{n-1},\ldots,\alpha_{1})$ after such a 
step. This will allow us to determine the modified 
bound after each extension very efficiently.

\begin{itemize}
  \item Note that a parallel extension reduces the total set 
  of unused available
  vertices and the set of unused available vertices 
  within the cover by 1. This change
  results in the following modification of the vector
  $(\alpha_{n},\alpha_{n-1},\ldots,\alpha_{1})$: 
  Let $h$ be the maximum value of
  $i$ such that $\alpha_i>0$, then one can 
  compute the modified vector
  $\vec{\alpha}'=(\alpha_{n}',\alpha_{n-1}',\ldots,\alpha_{1}')$ 
  by
  $$\alpha_{k}'\ \ = \ \ \left\{\begin{array}[c]{ll}
  \alpha_{k} & \text{for }k\not\in\{h,h-1\}\\
  \alpha_{k}-1 & \text{for }k=h\\
  \alpha_{k}+1 & \text{for }k=h-1\ .
  \end{array}\right.$$
  \item Note that a chain extension reduces the total set of 
  unused available vertices
  by 2 and the set of unused available vertices within the 
  cover by 1. This change
  results in the following modification of the vector
  $(\alpha_{n},\alpha_{n-1},\ldots,\alpha_{1})$:

  \begin{itemize}
    \item
    Let $h$ be the maximum 
    value of $i$ such that $\alpha_i>0$
    and let $g$ be the minimum value of $i$ such that $\alpha_i>0$.
    Note that either $h=g$ or $h=g+1$.
    \item To compute the modified vector 
    $\vec{\alpha}'=(\alpha_{n}',\alpha_{n-1}',\ldots,\alpha_{1}')$ 
    we have to reduce $\alpha_g$ by 1 and shift $g-2$ ones to 
    a higher position. This can be done by the following algorithm:

    \begin{enumerate}
      \item Choose the temporary values $t_i=\alpha_i$ for all $i\not=g$, and
      $t_g=\alpha_g-1$.
      \item Set $s=g-2$.
      \item If $t_g> s$ set $t_g=t_g-s$, $t_{g+1}=t_{g+1}+s$, $s=0$.
      \item If $t_g\le s$ set $t_{g+1}=t_{g+1}+t_g$, $t_g=0$, $s=s-t_g$, $g=g+1$.
      \item If $g=h$ and $s\ge t_g>0$ set $d=\lfloor s/t_g\rfloor$, $t_{g+d}=t_g$, $t_g=0$, 
      $s=s-d\cdot t_g$, $g=g+d$, $h=h+d$.
      \item If $g=h$ and $0<s<t_g$ set $t_{g+1}=s$, $t_{g}=t_g-s$, $h=h+1$, $s=0$.
      \item Finally set $\alpha_{i}'=t_i$ for all $i$.
    \end{enumerate}

    \item Since during this computation at most two temporary variables $t_k$
    are greater than 0, we can implement the functionality of this algorithm by
    using only 6 variables.

  \end{itemize}

  \item Since there are at most 2 entries of $\vec{\alpha}$ (and
  $\vec{\alpha}'$) that are different from 0, we only need to store these two
  values and the corresponding coordinates. Thus, four variables suffice to
  implement the lexicographical minimum vector $\vec{\alpha}$.
\end{itemize}

\section{Application Example: Generating Hard Instances for Benchmarks}
\label{sec:results}

One application of our \Irreducible graphs and generation processes based on
them is the construction of graphs that are difficult to solve for algorithms
focused on graph problems, e.g. the minimum vertex cover problem. Therefore, 
so-generated graphs can have the potential use of serving as benchmarks for
challenging current and future algorithms. This section shows the results of an
examplary test on a state-of-the-art heuristic and Naive Greedy as reference on
our graphs.

\textbf{Generating Hard Instances.}
To find graphs that are hard to solve within acceptable time, we
used a restricted version of the randomized graph generator described in
Definition~\ref{def:random_processes}. By generating several graphs, we could
verify that choosing two \Irreducible graphs within the first phase by
$\cG_{\cB,\ell,m,n}^1$ is sufficient for finding hard graphs. Moreover, our
experiments led us to the intuition that hard graphs can be found if we choose
the cover size $\ell$ a little bit larger than $n/2$. Our results for the
modification of $\cG_{\cB,\ell,m,n}^1$ are given in the following:
\begin{itemize}
  \item We generate two randomized \Irreducible graphs from
  $\cB=\chaParExtP{K_3}{*}$, such that each of their accumulated minimal vertex
  cover size is approximatively $n/4$ (the process generates such a graph with
  minimum cover size of $n/4$ and stops its extension randomly).
  \item Thereupon, we add the remaining vertices to get a graph of size $n$.
  \item Then, we add random edges to obtain a graph with the desired number of
  edges, ensuring that no two vertices that are not in the optimal cover share
  an edge.
  \item Finally, we permute the order of the vertices.
\end{itemize} 
Throughout all steps, we keep track of one optimal vertex cover of the graph.

We analyzed the performance of the updated version of NuMVC (v2015.8), where
we used the original code provided by its
authors\footnote{\url{http://lcs.ios.ac.cn/\~caisw/MVC.html}}.

\begin{table}[h]
  \small
  \begin{center}
  \scalebox{0.8}{
    \begin{tabular}{rr@{\hspace{0.5em}}rrr@{\hspace{0.5em}}rrrr}
      \toprule
          \multicolumn{3}{c}{Graphs}
          & \multicolumn{3}{c}{Distance of NuMVC} 
          & \multicolumn{3}{c}{Distance of Naive Greedy} \\
        \cmidrule(lr@{0.5em}){1-3}
        \cmidrule(lr@{0.5em}){4-6}
        \cmidrule(lr@{0.5em}){7-9}
        $n(G)$ & $k$ & Avg. Opt.
          & \#Opt. & Avg. & Max. 
          & \#Opt. & Avg. & Max. \\
        \midrule
        1500 & 1.1 &  752.7 & 10 &   0.0 &   0 &  0 &  24.3 &  28 \\
        1500 & 1.3 &  811.1 & 10 &   0.0 &   0 &  0 &  10.5 &  25 \\
        1500 & 1.5 &  834.9 &  9 &   0.1 &   1 &  0 &  26.9 &  51 \\
        1500 & 1.7 &  793.1 &  9 &  11.3 & 113 &  0 & 154.6 & 188 \\
        1500 & 1.9 &  918.3 & 10 &   0.0 &   0 & 10 &   0.0 &   0 \\
        \midrule
        2500 & 1.1 & 1270.6 & 10 &   0.0 &   0 &  0 &  41.0 &  52 \\
        2500 & 1.3 & 1372.6 & 10 &   0.0 &   0 &  0 &  28.0 &  54 \\
        2500 & 1.5 & 1363.1 &  9 &   0.1 &   1 &  0 &  51.1 &  87 \\
        2500 & 1.7 & 1431.7 &  8 &  29.4 & 175 &  0 & 176.0 & 264 \\
        2500 & 1.9 & 1418.8 & 10 &   0.0 &   0 & 10 &   0.0 &   0 \\
        \midrule
        3500 & 1.1 & 1797.7 & 10 &   0.0 &   0 &  0 &  59.9 &  91 \\
        3500 & 1.3 & 1884.0 &  2 &   5.4 &  16 &  0 &  38.9 &  56 \\
        3500 & 1.5 & 1975.5 & 10 &   0.0 &   0 &  0 &  62.3 &  96 \\
        3500 & 1.7 & 2009.6 &  4 & 154.0 & 203 &  0 & 189.1 & 336 \\
        3500 & 1.9 & 1955.8 & 10 &   0.0 &   0 & 10 &   0.0 &   0 \\
      \bottomrule
    \end{tabular}
  }
  \end{center}
  \vspace{-1.0em}
  \caption{Results for running NuMVC and Naive Greedy for a maximum of
  \unit[$1500$]{seconds} per run on graphs constructed with our approach. $n(G)$
  and $m(G) = n^k$ denote the number of vertices and edges, respectively. We
  constructed ten graphs for each combination of $n$ and $k$.
  \#Opt. denotes the number of runs where the respective algorithm found the
  optimal solution. Avg. and Max. distance denote the average and maximal number
  of vertices above the optimum, respectively.}
  \label{tab:varying-density-1}
\end{table}

\textbf{Results for NuMVC.}
Tables~\ref{tab:varying-density-1} summarizes our
experimental results for NuMVC for varying number of vertices $n$ and number of
edges $m$, where $m = n^k$ is determined by a polynomial function on $n$.
The table presents the quality of the found vertex
cover, where the column \emph{Avg. Opt.} provides the average optimum cover
size, \emph{\#Opt.} the number of instances for which NuMVC found an optimal
solution, and the remaining two columns the average and maximal distance between
the vertex cover found by NuMVC and the minimum vertex cover. One can observe
that the worst behavior of NuMVC take place for $k=1.7$. For sparse or dense
graphs, i.e. $k=1.1$ or $k=1.9$, NuMVC always found an optimal solution. 
Our results lead to the observation that choosing $k$ between $1.5$ and
$1.7$ leads to hard instances.

\textbf{Results for Greedy.}
We also studied the Naive-Greedy algorithm on our graphs, which always adds a
vertex to its resulting cover from those with a maximum degree within the
remaining graph. The experiments show that the Naive-Greedy algorithm can
determine an optimal solution for dense graphs. As for NuMVC graphs, the choice
of $k = 1.7$ leads to graphs that seem to be infeasible to solve. So, one can
assume that the error of NuMVC results from that of the initial cover generated
by its internally used greedy algorithm.

\textbf{Comparison with Witzel Graphs.}
The effectiveness of a random benchmark generation process also relies on its
frequency of producing hard instances. One could, e.~g., preselect some vertices
$C$ and add edges randomly among them and from them to remaining vertices --
ensuring that every vertex of $C$ shares an edge with at least one vertex in
$\overline{C}$. This process is similar to that discussed above; though, the
probability that the generated instance is hard is low and unpredictable. We
experimentally constructed 50 instances from this process. None of them appeared
hard to solve for NuMVC, i.~e. in all cases, it required a single step to find a
solution. It is hard to determine if NuMVC always found an optimal cover, but it
is at least highly likely; one can see from Table
\ref{tab:varying-density-1}, that, whenever NuMVC took a single step to solve an
instance, it always returned an optimal solution. We compared our generation
process to that of Witzel graphs, which start from $n$ cliques of $x$ vertices
each and ``[...] connect random pairs of cliques by random
edges''~\cite{dharwadker:2006}, which resembles the BHOSLIB
process~\cite{xu:2007}. For the parameters used to generate BHOSLIB Graphs
see~\cite{xu:2005simple}.

We executed NuMVC for \unit[$1500$]{seconds} over $30$ random instances
consisting of $4000$ vertices each, generated by the Witzel Graph process, the
BHOSLIB process, and our benchmark generator. For the construction of Witzel
graphs, we chose $100$ cliques of $40$ vertices each, and added edges randomly
such that their total number was located between $550\,000$ and $850\,000$.
Analogously, we also chose a number of $m \in [550\,000, 850\,000]$ random edges
for our benchmark generator. For the BHOSLIB process, we chose $n = 100$,
$\alpha = 0.8$, $r = 0.8/(\ln(4) - \ln(3)) \approx 2.7808$, and $p = 0.25$,
which were used for the hardest instance \texttt{frb-100-40} according to the
BHOSLIB authors.
NuMVC found an optimal solution for all $30$ Witzel instances and $29$ of the
BHOSLIB graphs, whereas it returned a minimum cover in only $23$ of the $30$
cases of our benchmark generator. Detailed results can be seen in 
Table~\ref{tab:bhoslib-process}. In a future step, we will compare our generator with the
BHOSLIB process.

\begin{table*}[h]
  \small
  \begin{center}
    \begin{tabular}{lrrrrrr}
      \toprule 
        & 
        & \multicolumn{3}{c}{Distance of NuMVC} \\
        \cmidrule(lr){3-5}
      Generator        & Avg. Opt. VC 
        & \#Opt.       & Avg. & Max. 
        & Avg. \#Steps & Avg. Time (s) \\
      \midrule
        Our generator  & 2250.63 & 23 & 16.57 & 102 &  8364265.5 & 118.37 \\
        BHOSLIB Graphs & 3900.00 & 29 &  3.00 &   3 & 11465347.9 & 398,22 \\
        Witzel Graphs  & 3900.00 & 30 &  0.00 &   0 & 21821546.0 & 441.84 \\
      \bottomrule
    \end{tabular}
  \end{center}
  \vspace{-1.0em}
  \caption{Comparison of minimum vertex cover size and run time of 
  NuMVC for a maximum of \unit[$1500$]{seconds} on
  \unit[$30$]{graphs} each from our graph generator, BHOSLIB Graphs, 
  and Witzel Graphs.}
  \label{tab:bhoslib-process}
\end{table*}

\textbf{Generating Instances with Less Structure}
Finally, we tackled the question whether the structure that is given by the
\Irreducible graphs is necessary for generating hard instances. Therefore, we
generated random graphs by the following simplified generator $\cG_{n,m,n_c}$:
\begin{quote}
$\cG_{n,m,n_c}$ outputs a graph of $n$ vertices $V$ and $m$ edges $E$ such that
for a subset $V_C\subseteq V$ of size $n_C$ the edges are uniformly chosen from
the set of edges connecting a vertex of $V_C$ with an arbitrary other vertex of
$V$ (see Figure~\ref{fig:RandProc02}).
\end{quote}
Clearly, $n_C$ is an upper bound for the minimum vertex cover of graphs
generated by $\cG_{n,m,n_c}$.

\begin{figure}[h]
  \begin{center}
    \scalebox{.4}{
      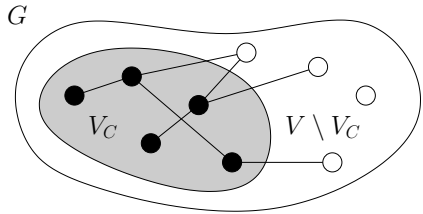
    }
  \end{center}
  \vspace{-1.0em}
  \caption{Illustration of a graph of $\cG_{n,m,n_c}$.}
  \label{fig:RandProc02}
\end{figure}

For our experiments, we chose $n=3500$ and $m=n^{1.7}$ and generated ten random
graphs for each value in $n_C\in\{2000, 2050, 2100, 2150, 2200\}$. For all
generated graphs, NuMVC found a solution latest after a single improvement step
that has been within the upper bound of $n_C$. We conclude from these results
that using the structure of \Irreducible graphs increases the probability of
finding hard instances. Moreover, knowing the exact size of a cover is highly
useful for determining the quality of the solutions determined by NuMVC.

\section{Conclusions}
\label{sec:conclusion}

This paper introduced the concept of \emph{\Irreducible} graphs for the vertex-%
cover problem, i.e. graphs which cannot be reduced wrt. their number of vertices
without reducing the minimum vertex cover. As starting points, they allow to
generate any graph within a random process such that we can determine the
optimal cover size of the resulting construction. The random process can be
parametrized by the desired number of vertices, cover size, as well as number of
edges. We introduced a set of extensions that allow us to generate a
sufficiently large subset of
\Irreducible graphs as the base for constructing graphs. We showed that our
extensions allow to render the generation process very efficiently. 

One specific application with apparent potential is the construction of graphs
with hidden solutions that are difficult to solve. Our experiments with a recent
heuristic NuMVC and Naive Greedy frequently underestimate the minimum vertex
cover by more than $100$ vertices for graphs generated with our process with
\unit[$3500$]{vertices}. By analyzing the behavior of NuMVC and Naive Greedy on
random graphs without a substructure like our \Irreducible graphs, we could
observe that these structures appear necessary for hard instances.
We will continue this research and publish the hardest graphs at our webside.

In general, \Irreducible graphs are interesting objects to study to deepen the
understanding of hard problems. Concerning our proposed generation process, if
we were aware of all \Irreducible instances, one could efficiently create all
graphs from our process. Though, it seems that the task of finding all critical
instances is far from complete. We have started the systematic search of small
\Irreducible instances and will continue will report on further findings in the
close future.

Furthermore, we started analyzing the quality of {\CirculantGraph}s for
benchmarking purposes. At an early stage, we found that instances with a small
size of internal cliques were still easily solvable for local-search algorithms
such as NuMVC, which shows that the verification of partial results is crucial.
We plan to continue our research also in the direction of constructing larger
\Irreducible graphs from {\CirculantGraph}s and will also report on those
findings.

{\bf Acknowledgement:} We would like to thank Gwena\"{e}l Joret and Ke Xu 
for their constructive comments on an earlier version of the paper.


\bibliographystyle{plain}
\bibliography{TR}


\newpage

\appendix
\section{Critical Graphs of Degree $6$}
\label{sec:degree-four}

\begin{table}[H]
  \centering
  \scriptsize
  \begin{tabular}{r p{140mm}}
    \toprule
      $n$ & Tuples $(i, j)$ \\
    \midrule
    $ 4$ &
    $( 2,  3)$  \\
    $ 5$ &
    $( 2,  3)$, $( 2,  4)$, $( 3,  4)$  \\
    $ 6$ &
    $( 2,  3)$, $( 3,  4)$  \\
    $ 7$ &
    $( 2,  3)$, $( 2,  4)$, $( 3,  5)$, $( 4,  5)$  \\
    $ 8$ &
    $( 2,  6)$, $( 2,  7)$, $( 6,  7)$  \\
    $10$ &
    $( 4,  5)$, $( 5,  6)$  \\
    \midrule
    $11$ &
    $( 2,  3)$, $( 2,  8)$, $( 2,  9)$, $( 2, 10)$, $( 3,  4)$, $( 3,  7)$, $( 3,  9)$, $( 4,  5)$, $( 4,  6)$, $( 4,  8)$, $( 5,  6)$, $( 5,  7)$, $( 5, 10)$, $( 6,  7)$, $( 6, 10)$, $( 7,  8)$, $( 8,  9)$, $( 9, 10)$  \\
    $13$ &
    $( 3,  4)$, $( 3,  9)$, $( 4, 10)$, $( 5,  8)$, $( 5, 12)$, $( 8, 12)$, $( 9, 10)$  \\
    $14$ &
    $( 2, 12)$, $( 2, 13)$, $( 6,  7)$, $( 7,  8)$, $(12, 13)$  \\
    $15$ &
    $( 2,  3)$, $( 2, 12)$, $( 3, 13)$, $( 6,  7)$, $( 6,  8)$, $( 7,  9)$, $( 8,  9)$, $(12, 13)$  \\
    $17$ &
    $( 2,  6)$, $( 2, 11)$, $( 2, 15)$, $( 2, 16)$, $( 3,  6)$, $( 3,  8)$, $( 3,  9)$, $( 3, 11)$, $( 4, 13)$, $( 4, 16)$, $( 6, 14)$, $( 6, 15)$, $( 8,  9)$, $( 8, 14)$, $( 8, 16)$, $( 9, 14)$, $( 9, 16)$, $(11, 14)$, $(11, 15)$, $(13, 16)$, $(15, 16)$  \\
    $18$ &
    $( 3,  4)$, $( 3, 14)$, $( 4, 15)$, $( 8,  9)$, $( 9, 10)$, $(14, 15)$  \\
    $19$ &
    $( 2,  3)$, $( 2, 16)$, $( 3, 17)$, $( 6,  7)$, $( 6, 12)$, $( 7,  8)$, $( 7, 11)$, $( 7, 13)$, $( 8,  9)$, $( 8, 10)$, $( 8, 12)$, $( 9, 11)$, $(10, 11)$, $(11, 12)$, $(12, 13)$, $(16, 17)$  \\
    $20$ &
    $( 2, 18)$, $( 2, 19)$, $( 3,  4)$, $( 3, 16)$, $( 4, 17)$, $( 7,  8)$, $( 7, 12)$, $( 8, 13)$, $(12, 13)$, $(16, 17)$, $(18, 19)$  \\
    \midrule
    $22$ &
    $( 4, 18)$, $(10, 11)$, $(11, 12)$  \\
    $23$ &
    $( 2,  3)$, $( 2,  9)$, $( 2, 14)$, $( 2, 20)$, $( 5, 10)$, $( 5, 13)$, $( 7,  8)$, $( 7, 11)$, $( 7, 12)$, $( 7, 15)$, $( 7, 16)$, $( 8, 16)$, $(10, 11)$, $(10, 12)$, $(10, 13)$, $(10, 18)$, $(11, 12)$, $(11, 13)$, $(11, 16)$, $(12, 13)$, $(12, 16)$, $(13, 18)$, $(15, 16)$  \\
    $24$ &
    $( 2,  6)$, $( 2, 18)$  \\
    $25$ &
    $( 3,  4)$, $( 6,  7)$, $( 6, 18)$, $( 7,  8)$, $( 7, 17)$, $( 7, 18)$, $( 7, 19)$, $( 8, 18)$, $( 9, 10)$, $( 9, 15)$, $(10, 11)$, $(10, 14)$, $(10, 16)$, $(11, 15)$, $(14, 15)$, $(15, 16)$, $(17, 18)$, $(18, 19)$  \\
    $26$ &
    $( 6, 10)$, $( 6, 16)$, $( 7,  8)$, $( 7, 18)$, $( 8, 19)$, $(10, 11)$, $(10, 15)$, $(10, 20)$, $(11, 16)$, $(12, 13)$, $(13, 14)$, $(15, 16)$, $(16, 20)$, $(18, 19)$  \\
    $27$ &
    $( 2,  3)$, $( 3,  4)$, $( 6,  7)$, $( 6, 20)$, $( 7, 20)$, $(12, 13)$, $(12, 14)$, $(13, 15)$, $(14, 15)$  \\
    $28$ &
    $( 3,  8)$, $( 3, 20)$, $( 5, 14)$, $( 6, 10)$, $( 6, 18)$, $( 7, 12)$, $( 7, 16)$, $( 9, 12)$, $( 9, 16)$, $(10, 18)$, $(11, 14)$, $(12, 19)$, $(14, 17)$, $(16, 19)$  \\
    $29$ &
    $( 2, 12)$, $( 2, 17)$, $( 5, 12)$, $( 5, 13)$, $( 5, 16)$, $( 5, 17)$, $( 6,  9)$, $( 6, 14)$, $( 6, 15)$, $( 6, 20)$, $( 9, 13)$, $( 9, 16)$, $(12, 17)$, $(13, 20)$, $(14, 15)$, $(16, 20)$  \\
    $30$ &
    $( 6,  7)$, $(12, 13)$, $(12, 17)$, $(13, 18)$, $(14, 15)$, $(15, 16)$, $(17, 18)$  \\
    \midrule
    $31$ &
    $( 2,  3)$, $( 2,  6)$, $( 3, 15)$, $( 3, 16)$, $( 5,  6)$, $( 5, 10)$, $(10, 11)$, $(10, 20)$, $(14, 15)$, $(14, 16)$, $(15, 17)$, $(16, 17)$  \\
    $32$ &
    $( 3,  4)$, $( 6,  7)$, $( 6, 12)$, $( 6, 20)$, $( 9, 10)$, $(11, 12)$, $(11, 20)$, $(12, 20)$  \\
    $33$ &
    $( 2,  9)$, $( 6,  7)$, $(12, 16)$, $(12, 17)$, $(14, 15)$, $(14, 18)$, $(15, 18)$, $(15, 19)$, $(18, 19)$  \\
    $34$ &
    $( 3,  4)$, $( 6,  9)$, $(10, 11)$, $(12, 15)$, $(12, 19)$, $(16, 17)$, $(17, 18)$  \\
    $35$ &
    $( 2,  3)$, $( 2, 15)$, $( 2, 20)$, $( 5, 16)$, $( 5, 19)$, $( 7,  8)$, $(10, 17)$, $(10, 18)$, $(11, 12)$, $(11, 15)$, $(11, 20)$, $(13, 14)$, $(16, 17)$, $(16, 18)$, $(17, 18)$, $(17, 19)$, $(18, 19)$  \\
    $36$ &
    $(15, 16)$, $(15, 20)$  \\
    $37$ &
    $(10, 11)$  \\
    $38$ &
    $( 2,  6)$, $( 6,  7)$, $( 8, 16)$, $(10, 11)$, $(18, 19)$, $(19, 20)$  \\
    $39$ &
    $( 2,  3)$, $( 3,  4)$, $( 3,  8)$, $( 3, 17)$, $( 5, 15)$, $( 9, 10)$, $( 9, 16)$, $(15, 16)$, $(17, 18)$, $(18, 19)$, $(18, 20)$  \\
    $40$ &
    $( 5, 12)$, $( 6,  7)$, $(10, 14)$, $(10, 18)$, $(12, 18)$, $(14, 15)$, $(15, 16)$, $(17, 18)$  \\
    \midrule
    $41$ &
    $( 2, 18)$, $( 3,  4)$, $( 3, 15)$, $( 5, 14)$, $( 5, 19)$, $( 8, 11)$, $( 8, 12)$, $( 9, 16)$, $( 9, 20)$, $(10, 11)$, $(13, 17)$, $(14, 15)$  \\
    $42$ &
    $( 2, 12)$, $( 6, 10)$, $( 6, 19)$, $(11, 18)$  \\
    $43$ &
    $( 2,  3)$, $( 2,  9)$, $( 5, 19)$, $( 6,  7)$, $( 8, 12)$, $(14, 15)$, $(15, 18)$, $(16, 20)$  \\
    $44$ &
    $( 5, 12)$, $( 6, 10)$, $( 8, 11)$, $( 8, 18)$, $( 9, 20)$  \\
    $45$ &
    $( 2,  6)$, $( 6,  7)$, $( 9, 10)$, $(12, 13)$  \\
    $46$ &
    $( 3,  4)$, $( 5, 12)$, $( 6,  7)$, $( 9, 16)$, $(10, 11)$, $(13, 14)$, $(14, 15)$, $(16, 17)$, $(18, 19)$  \\
    $47$ &
    $( 2,  3)$, $( 5, 17)$, $( 5, 19)$, $( 6, 19)$, $( 8, 11)$, $( 9, 10)$, $( 9, 18)$, $(10, 11)$, $(13, 14)$, $(14, 15)$, $(15, 16)$, $(15, 19)$, $(17, 18)$  \\
    $48$ &
    $( 3,  4)$  \\
    $49$ &
    $(18, 19)$  \\
    $50$ &
    $( 3,  8)$, $( 5, 14)$, $( 8, 12)$, $( 9, 10)$, $(10, 11)$, $(13, 14)$, $(14, 17)$  \\
    \midrule
    $51$ &
    $( 2,  3)$, $( 2, 15)$, $( 6,  7)$  \\
    $52$ &
    $( 2,  6)$, $( 5, 16)$, $(14, 18)$  \\
    $53$ &
    $( 2,  9)$, $( 3,  4)$, $( 4, 13)$, $( 4, 16)$, $( 6,  7)$, $( 6, 12)$, $( 8, 11)$, $( 8, 20)$, $( 9, 10)$, $(10, 11)$, $(10, 13)$, $(13, 14)$, $(15, 16)$, $(16, 17)$, $(18, 19)$  \\
    $55$ &
    $( 2,  3)$, $( 2, 12)$, $( 3,  4)$, $(10, 11)$, $(12, 13)$, $(13, 14)$, $(16, 17)$, $(17, 18)$, $(18, 19)$  \\
    $56$ &
    $( 5, 14)$, $( 5, 16)$, $( 5, 18)$, $( 6,  7)$, $( 6, 17)$, $( 8, 11)$, $( 8, 18)$, $( 9, 10)$, $(11, 14)$, $(14, 17)$, $(14, 20)$  \\
    $57$ &
    $( 5, 12)$, $(12, 16)$, $(15, 16)$  \\
    $58$ &
    $( 4, 16)$, $( 6,  7)$, $( 6, 10)$, $( 8, 11)$, $( 9, 10)$, $(10, 16)$, $(13, 14)$  \\
    $59$ &
    $( 2,  3)$, $( 2,  6)$, $( 6,  7)$, $(10, 11)$, $(10, 20)$, $(11, 14)$, $(12, 16)$, $(12, 18)$, $(12, 19)$, $(16, 17)$, $(19, 20)$  \\
    $60$ &
    $( 2, 18)$, $( 3,  4)$, $( 5, 14)$, $( 5, 16)$, $( 6, 10)$, $( 8, 11)$, $( 8, 14)$  \\
    \bottomrule
  \end{tabular}
  \caption{Critical circulant graphs of degree $6$. Within this table 
    the value 1 is omitted within each entry.}
  \label{tab:critical-circulant-graphs-2}
\end{table}


\end{document}